\documentclass[utf8]{article}

\usepackage{enumerate}
\usepackage{titlesec}
\usepackage{indentfirst}
\usepackage{booktabs}

\usepackage{multirow}
\usepackage{geometry}
\usepackage{setspace}

\usepackage{graphicx}
\usepackage{float}
\usepackage{subfigure}
\usepackage{color}

\usepackage{amsthm,amsmath,amssymb}

\usepackage[round, comma]{natbib}
\usepackage{appendix}

\newtheorem{theorem}{Theorem}
\newtheorem{lemma}[theorem]{Lemma}

\newtheorem{corollary}[theorem]{Corollary}
\newtheorem{remark}{Remark}
\newtheorem{assumption}{Assumption}

\newcommand{\E}[1]{\mathbb{E}\left[#1\right]}

\newcommand{\R}{\mathbb{R}}

\newcommand{\argmin}{\mathop{\arg\min}}
\newcommand{\abs}[1]{\left\vert#1\right\vert}
\newcommand{\br}[1]{\left(#1\right)}
\newcommand{\N}{\mathbb{N}}

\newcommand{\B}{\boldsymbol}
\newcommand{\BmR}{\boldsymbol{\mathcal{R}}}

\newcommand{\D}{\mathcal{D}}

\title{Global Bias-Corrected Divide-and-Conquer by Quantile-Matched Composite for General Nonparametric Regressions}
\author{Yan Chen and Lu Lin\footnote{The corresponding
author. Email: linlu@sdu.edu.cn. The research was
supported by National Key R\&D Program of China (2018YFA0703900) and NNSF project (11971265) of China.}
\\
 $^1$Zhongtai Securities Institute for Financial Studies, Shandong University, Jinan, China
}

\begin{document}

\maketitle

\begin{abstract}\baselineskip=16pt

The issues of bias-correction and robustness are crucial in the strategy of divide-and-conquer (DC), especially for asymmetric nonparametric models with massive data. It is known that quantile-based methods can achieve the robustness, but the quantile estimation for nonparametric regression has non-ignorable bias when the error distribution is asymmetric.
This paper explores a global bias-corrected DC by quantile-matched composite for nonparametric regressions with general error distributions. The proposed strategies can achieve the bias-correction and robustness, simultaneously.
Unlike common DC quantile estimations that use an identical quantile level to construct a local estimator by each local machine, in the new methodologies, the local estimators are obtained at various quantile levels for different data batches,
and then the global estimator is elaborately constructed as a weighted sum of the local estimators. In the weighted sum, the weights and quantile levels are well-matched such that the bias of the global estimator is corrected significantly, especially for the case where the error distribution is asymmetric.
Based on the asymptotic properties of the global estimator, the optimal weights are attained, and the corresponding algorithms are then suggested. The behaviors of the new methods are further illustrated by various numerical examples from simulation experiments and real data analyses. Compared with the competitors, the new methods have the favorable features of estimation accuracy, robustness, applicability and  computational efficiency.

\end{abstract}

\tableofcontents

\baselineskip=20pt

\section{Introduction}
\subsection{Problem setup and article frame}

A commonly used strategy for analyzing massive data set is the well-known method of divide-and-conquer, denoted by DC for short. From the perspective of statistical estimation,
a main challenge in DC algorithms is how to sufficiently reduce the estimation bias, especially for the case when the bias of local estimators is ineradicable, such as the bias from the nonparametric and penalty-based methods.
The second challenge is how to construct robust strategy since
outliers, fat-tailed features and heterogeneity are more likely to hide in these massive data, and it is usually hard to detect or address them because of the structural complexity of massive data. Although quantile-based methodologies can achieve the robustness, the quantile estimation for nonparametric models has non-ignorable bias when the error distribution is asymmetric, and such a bias is irreducible as well in the procedure of common DC algorithms. We then are interested in the following issue:
\begin{itemize}\item {\it Develop quantile-based algorithms such that the final bias of the global DC estimator can be significantly corrected. Then, our goal is to achieve the bias-correction and robustness, simultaneously. }\end{itemize}
It is known that common DC quantile estimations use an identical quantile level to construct a local estimator by each local data set. Unlike common DC, the main feature of our technique is what follows:
\begin{itemize}\item {\it It is a multilevel quantile-matched composite methodology. Specifically,
the local estimators are obtained at various quantile levels for different data batches,
and then the global estimator is elaborately constructed as a weighted sum of the local estimators, in which the weights and quantile levels are well-matched such that the bias of the global estimator is corrected significantly.}\end{itemize}
The paper is then organized in the following way. In the remainder of this section, some related works
on DC are reviewed, and our
main contributions are summarized.
In Section~\ref{sec_meth}, the main idea of our methodologies is introduced.
In Section~\ref{sec_theo}, the asymptotic properties of the proposed method are established, the optimal selections of the weights, quantile levels and bandwidths are discussed deeply.
In Section~\ref{sec_praimple}, the estimations of unknown parameters needed in computing the optimal weights and bandwidths are suggested.
Section~\ref{sec_simu} contains comprehensive simulation studies and real data analyses to further illustrate our theoretical findings.
Some conditions and technical proofs are deferred to the Appendix.

\subsection{Related works}

In the procedure of DC,
the local results (e.g., the local estimators of a parameter) are obtained by each batch of data, and then the global result is constructed by aggregating these local results.
See e.g. \cite{Guha2008, Lin2011, Xie_2014, Song2015, Scott2016, Pillonetto2019, Chen2019} and the references therein.
A variety of DC strategies have been widely used in statistics. The related works include but are not limited to the DC expression for linear model of \cite{2006Regression, Lin2011} and \cite{Schifano2016}, the density estimation of \cite{Li2013}, the parametric regression estimation of \cite{Chen2014}, and \cite{Zhang2015Divide}, the high-dimensional parametric regression estimation of \cite{Lee2017Communication}, the semi-parametric regression estimation of \cite{Zhao2016partially}, the M-estimator of \cite{Shi2018massive},
the distributed testing and estimation of \cite{Battey2018Distributed},
the quantile regression processes of \cite{Volgushev2019Distributed},
the optimal subdata selection of \cite{Wang2019Information},
the sparse logistic regression of \cite{divide2020Hong},
the communication-efficient composite quantile regression of \cite{WANG2021Robust},
and the distributed statistical inference of \cite{Chen2021Distributed}.

It is known that the convergence rate of the variance of a typical estimator is usually of the order proportional to $\sqrt{n}$ with $n$ the sample size.
When dealing with massive data sets, the estimation variance can be quite small even if the used estimation method is less efficient.
Thus, when evaluating the estimation method for analyzing massive data sets, variance-reduction is not very significant.
Contrarily, the bias-correction is crucial in the framework of DC, it is because
the accumulated bias of local estimators is typically non-negligible with the increasing of the number $m$ of data batches. In the worst case, the global estimator would be even divergent as $m $ tends to infinity if the bias of the local estimators is not small enough.
Due to the bias of local estimators, the number of batches is usually restricted, to achieve the standard convergence rate for the global estimator. This imposes great restrictions on the DC methodologies \citep[see, e.g.][]{Chen2019, lin2019global, Chen2021Distributed}.
To address this issue, the most common procedures use iterative algorithm \citep[see, e.g.][]{Jialei2017Efficient} and local bias-correction \cite[see, e.g.][]{Lee2017Communication, Lian2019Projected, Li2022Score} to reduce the bias of
local estimators and then to control the bias of the global estimator.
However, the iterative algorithm and the bias-correction for local estimators are computationally complex, and the resulting bias-correction for global estimation is not sufficient usually.

Apart from the bias-correction, the robustness plays an important role the area of massive data. It is known that outliers,
fat-tailed and heterogenous distributions often appear in the environment of massive data. Because of the structural complexity of massive data sets, it is usually difficult to recognize and address these abnormal phenomena. Thus the robust strategies are particularly desirable for analyzing massive data.
Quantile-based methods are frequently used in classical statistics because of their robustness and relative efficiency.
For nonparametric regression, the local linear quantile regression \citep{Fan1994Robust, Welsh1996Robust} is a robust alternative to local linear regression \citep{Fan1996Local}.
Based on this method, a kind of composite quantile regressions (CQRs) \citep[e.g. ][]{Zou2008Composite, Kai2010, Sun2013Weighted, LinComposite2019, Gu2020Sparse, WANG2021Robust} have been proposed for parametric and nonparametric models.
The CQRs enjoys great advantages in terms of estimation efficiency whether the
variance of error is finite or not.
In \cite{Kai2010}, the proposed local CQR could even be much more efficient than local linear least squares in some cases of non-normal and symmetric error distributions.

Despite these advantages, in nonparametric regressions, the aforementioned CQR methods deeply rely on the assumption of symmetric random errors.
If the errors were asymmetric, these methods could become invalid, since the estimation would suffer a non-vanishing bias; see the proof of Theorem 1 in \cite{Kai2010}. To correct the non-vanishing bias for the case of symmetric conditions,
\cite{Sun2013Weighted} introduce a weighted local linear CQR (WCQR).
The WCQR method contains local linear CQR as a special case, and moreover, it not only inherits good properties that local linear CQR owns for symmetric errors, but also is applicable to asymmetric error distributions.
The main idea and desirable features of WCQR motivate us to put forward a robust estimation for nonparametric regressions with massive data sets.

\subsection{Our main contributions}

As stated before, although there have already been numerous studies on the estimation method in the framework of DC, few of them focus on bias-correction and robustness, simultaneously.
Hence it is particularly desirable to explore new strategies with  the features of bias-correction and robustness.
In this paper, our goal is then to achieve the two objectives by quantile-based method in the framework of DC for nonparametric regression with general error distributions.
\begin{itemize}
    \item
{\it Unlike common quantile methods, ours is a multilevel quantile-matched composite.}\end{itemize}
Firstly, the local estimators are constructed by quantile estimation respectively at various quantile levels from each data batch.
Secondly, the global estimator is designed  as a weighted sum of the local estimators. As a consequence, when the weights and quantile levels are well-matched, the bias of the global estimator is corrected significantly.
Moreover, the optimal weights are obtained by minimizing the asymptotic variance of the global estimator, and the bandwidths are selected under criterion of the optimal convergence rate.
Theoretically, the asymptotic properties of our global estimator are established.
The behavior of the method is further illustrated by various numerical examples from simulation experiments and real data analyses.

Compared with the competitors, our method has the following main virtues:
\begin{itemize}
    \item[1)] {\it Robustness}. Thanks to the robust feature of the quantile estimation, our method enjoys the robustness compared with the common methods such as ordinary least squares (OLS).
    \item[2)] {\it Bias-correction}.
    The bias of the global estimator is significantly reduced by
    the quantile-matched composite.
    \item[3)] {\it General applicability}. Because of fully use of  weights and quantile levels, the method is adaptive to  symmetric or asymmetric distributions, homogeneous or heterogenous models.
    \item[4)] {\it Computationally efficiency}. The procedure of constructing global estimator is an one-shot aggregation via weighted averaging, without need for multiple rounds of iterative algorithms.
\end{itemize}

\subsection{Notations}

Denote by $f_Z(\cdot)$ and $F_Z(\cdot)$ the the probability density function (PDF) and the cumulative distribution function (CDF) of a random variable $Z$, respectively.
Let $\N$ be the set of natural numbers.
Denote by $\B{I}_m$ and $\B{1}_m$ the $m\times m$-dimensional identity matrix and the $m\times m$-dimensional matrix with all the elements equal to $1$, respectively.
For the sequences $\{a_n\}_{n=1}^{\infty}$ and $\{b_n\}_{n=1}^{\infty}$, say $a_n = \Omega(b_n)$ if $a_n = O(b_n)$ and $b_n = O(a_n)$, as $n \to \infty$.
For a vector $\B{v} = \br{v_i}_{i=1}^d = \br{v_1, \cdots, v_d}^\top$, denote $\|\B{v}\|_q = \br{\sum_{i=1}^d v_i^q}^{1/q}$ for $q = 1, 2$, and denote $\|\B{v}\|_{\infty} = \max_{1 \leq i \leq d}\abs{v_i}$.
For vectors $\B{u}, \B{v} \in \R^d$, denote
$\B{v}_1 \otimes \B{v}_2 = \br{u_1, \cdots, u_d, v_1, \cdots, v_d}^{\top} \in \R^{2d}$.
For vectors $\B{v}_1, \cdots, \B{v}_q \in \R^d$, write $\otimes_{i=1}^q \B{v}_i = \B{v}_1 \otimes \cdots \otimes \B{v}_q$.
For a matrix $\B{M}$, the minimum and maximum eigenvalues of $\B{M}$ are denoted by $\lambda_{\min}(\B{M})$ and $\lambda_{\max}(\B{M})$, respectively.
For matrices $\B{S}_1, \cdots, \B{S}_q$, define the block matrix by
$$\text{diag}\br{\B{S}_1, \cdots, \B{S}_q} = \begin{bmatrix}
    \B{S}_1&  &   & \B{0}\\
    &  \B{S}_2&   &      \\
    &    &    \ddots&      \\
    \B{0}&  &  & \B{S}_q
   \end{bmatrix}.
$$

\section{Methodology}\label{sec_meth}
\subsection{Model and data sets}

We consider the following general nonparametric regression model:
\begin{equation}\label{eq_model}
    Y = m(X) + \sigma(X)\varepsilon.
\end{equation}
Here, for simplicity, $Y$ and $X$ are supposed respectively to be scalar response variable and covariate, with $X\in [0,1]$. In the model, $m(\cdot): [0,1] \to \R$ is a smooth regression function to estimate, $\sigma(\cdot)\in R^+$ is the standard derivation of the error term, $\varepsilon$ is the random error independent of $X$, and satisfies $\E{\varepsilon} = 0$ and $\text{Var}[\varepsilon] = 1$.

Let $\D = \left\{(X_{ij}, Y_{ij}): j=1, \cdots, n_i, i = 1, \cdots, m\right\}$ be the massive data set of independent and identically distributed (i.i.d.) observations of $(X, Y)$.
For estimation from massive data, the full data set $\D$ is divided into $m$ data batches with the $i$-th data batch as $\D_i = \left\{X_{ij}, Y_{ij}\right\}_{j=1}^{n_i}$.
Thus, the size of data set $\D_i$ is $n_i$, and the total number of samples is $n = \sum_{i=1}^m n_i$.

\subsection{Quantile-matched composite estimator and bias correction}

Denote by $Q_Y(\tau|x)$ the conditional $100\tau\%$ quantile of $Y$ given $X=x\in [0,1]$ and $\tau \in (0, 1)$, i.e. $Q_Y(\tau|x) = \inf\left\{t: \mathbb{P}\br{Y \leq t|X=x} \geq \tau \right\}$.
For $i = 1, \cdots, m$, the standard local linear quantile estimator for $Q_{Y}(\tau|x)$ from the $i$-th data batch is given by
\begin{equation}\label{eq_hatm}
    \widehat{m}_i(x; \tau, h) = \widehat{a} \; \text{ with } \; (\widehat{a}, \widehat{b})= \argmin_{(a, b)\in \R\times \R} \sum_{j=1}^{n_i} \rho_{\tau}\left(Y_{ij} - a - b(X_{ij} - x)\right) K_{h}\left(X_{ij} - x\right),
\end{equation}
where $\rho_{\tau}(u) = \br{\tau-I(u<0)}u$, the check function, and $K_{h}(\cdot)= \frac{1}{h}K(\frac{\cdot}{h})$ with $K(\cdot)$ a kernel function and $h$ a smoothing parameter or bandwidth. As stated in Introduction, all the local estimators $\widehat{m}_i(x; \tau, h),i=1,\cdots,m$, have a non-vanishing bias if the error distribution is asymmetric.

To unbiasedly estimate the regression function $m(\cdot)$ in the framework of massive data, we design the following DC strategy:
\begin{enumerate}
    \item On each data batch $\D_i$ and quantile level $\tau_{ij}\in (0,1)$, $i = 1, \cdots, m,j=1,\cdots,J$, the $i$-th machine computes $J$ local quantile estimators $\left\{\widehat{m}_i \br{x; \tau_{ij}, h_{ij}}\right\}_{j=1}^J$ by \eqref{eq_hatm}, where $\tau_{ij}$ and $h_{ij}$ are respectively artificially selected quantile levels and bandwidths specified later;
    \item Each machine send $J$ local quantile estimators $\left\{\widehat{m}_i \br{x; \tau_{ij}, h_{ij}}\right\}_{j=1}^J$ to center machine, and the center machine constructs the global estimator by weighted averaging all the local estimators as
    \begin{equation}\label{nonp_glo_esti}
        \widehat{m}(x) = \sum_{i=1}^m \sum_{j=1}^J \omega_{ij} \widehat{m}_i \br{x; \tau_{ij}, h_{ij}},
    \end{equation}
    where $\omega_{ij}$ are artificially selected weights specified later;
    \item In the above two steps, the weights $\omega_{ij}$ and quantile levels $\tau_{ij}$ are well-matched via eliminating the estimation bias and minimizing the estimation variance of $\widehat{m}(\cdot)$ asymptotically (the details will be given later).
\end{enumerate}


Unlike common DC, the above is a new DC by multilevel quantile-matched composite.
To determine weights and quantile levels that guarantee the required properties of the global estimator, we first discuss the asymptotic representation of $\widehat{m}(x)$.
Under some regularity conditions \citep[see, e.g.][]{Hong2003, Chaudhuri1991Nonparametric, Guerre2012}, following local Bahadur representation holds:
\begin{equation}\label{eq_asy_glomx}
    \widehat{m}(x) = \sum_{i=1}^m \sum_{j=1}^J \omega_{ij} m(x) + \mathbb{I}_1 + \mathbb{I}_2 + \mathbb{I}_3 + \mathbb{I}_4,
\end{equation}
where
\begin{equation*}
    \mathbb{I}_1 = \sum_{i=1}^m \sum_{j=1}^J \omega_{ij} F_{\varepsilon}^{-1}(\tau_{ij})\sigma(x), \quad \mathbb{I}_2 = \sum_{i=1}^m\sum_{j=1}^J \omega_{ij} h_{ij}^2 \beta(x, \tau_{ij}),
\end{equation*}
\begin{equation*}
    \mathbb{I}_3 =  \sum_{i=1}^m\sum_{j=1}^J \omega_{ij}\phi(\tau_{ij}, h_{ij}, \D_i), \quad \mathbb{I}_4 = \sum_{i=1}^m\sum_{j=1}^J \omega_{ij} R^B_i(n_i, h_{ij}, \tau_{ij})
\end{equation*}
with
\begin{equation}\label{eq_defbeta}
    \beta(x, \tau_{ij}) = \frac{\mu_2}{2}\br{m^{\prime\prime}(x) + \sigma^{\prime\prime}(x) F_{\varepsilon}^{-1}(\tau_{ij})}, \quad \mu_2 = \int_{\R} v^2 K(v) dv.
\end{equation}
Here $\phi(\tau_{ij}, h_{ij}, \D_i)$ is given in \eqref{eq_phi} in the next section, and satisfies
\begin{equation*}
    \E{\phi(\tau_{ij}, h_{ij}, \D_i)} = 0, \quad \phi(\tau_{ij}, h_{ij}, \D_i) = O_p(n_i^{-1/2}),
\end{equation*}
and $R^B_i(n_i, \tau_{ij}, h_{ij})$ are the remainder terms satisfying
\begin{equation*}
    R^B_i(n_i, \tau_{ij}, h_{ij}) = O_p\left(\left(\frac{\log n_i}{n_i h_{ij}}\right)^{-3 / 4}\right) + o(h_{ij}^2).
\end{equation*}
In the asymptotic expansion \eqref{eq_asy_glomx}, the first term is equal to $m(x)$ if weights fulfill $\sum_{i=1}^m \sum_{j=1}^J \omega_{ij} = 1$.
Then the estimation error is determined by the terms $\mathbb{I}_1, \cdots, \mathbb{I}_4$.
Under some standard conditions (see, Theorem~\ref{them_asyproperty}),
$\mathbb{I}_4$ is an infinitesimal of higher order than the others; $\mathbb{I}_3$ has zero mean and characterizes the asymptotic variance of $\widehat{m}(x)$;
$\mathbb{I}_2$ is the estimation bias related with the bandwidths;
$\mathbb{I}_1$ would be a non-vanishing bias unless the weights were well-chosen.

Focusing on bias-correction, we introduce the following conditions on the weights and quantile levels:
\begin{equation}\label{eq_nonbias}
    \left\{\begin{aligned}
        &\sum_{i=1}^m \sum_{j=1}^J \omega_{ij} = 1, \\
        &\sum_{i=1}^m \sum_{j=1}^J \omega_{ij} F_{\varepsilon}^{-1}(\tau_{ij}) = 0.
    \end{aligned}\right.
\end{equation}
In \eqref{eq_nonbias}, the first condition is standard for weight selection, and the second one eliminates the non-vanishing bias $\mathbb{I}_1$.

\begin{remark}
    The global estimator given by \eqref{nonp_glo_esti} contains the  local linear CQR estimator \citep[e.g.][]{Kai2010} as a special case,
    in which the uniform weights are used.
    However, the uniform-weight estimator relies on the assumption of symmetric random errors.
    The symmetry ensures the validity of the second equation in \eqref{eq_nonbias}, so as to eliminate the non-vanishing bias $\mathbb{I}_1$ in the asymptotic expansion~\eqref{eq_asy_glomx} \citep{Sun2013Weighted}.
    Without this assumption, the uniform-weight estimator may suffer from a non-negligible bias.
\end{remark}

The candidate weights and quantile levels satisfying \eqref{eq_nonbias} are not unique.
Additionally, as shown in Theorem~\ref{them_asyproperty} in the next section, the weights $\omega_{ij}$, the quantile levels $\tau_{ij}$ and the bandwidths $h_{ij}$ have a significant influence on the feature of the asymptotic distribution of $\widehat{m}(x)$.
This suggests that we can select the best ones by optimizing an objective subject to \eqref{eq_nonbias} and, consequently, obtain an optimal global estimator.

\subsection{The optimal weights}

The optimal weights can be obtained by minimizing the asymptotic variance of $\widehat{m}(x)$ under constraint  \eqref{eq_nonbias}.
For convenience, denote the vectors of weights, quantile levels, and bandwidths by $\B{\omega}= \otimes_{i=1}^m \B{\omega}_i$, $\B{\tau}= \otimes_{i=1}^m \B{\tau}_i$ and $\B{h}= \otimes_{i=1}^m \B{h}_i$,
respectively, where $\B{\omega}_i = \br{\omega_{ij}}_{j=1}^J$, $\B{\tau}_i= \br{\tau_{ij}}_{j=1}^J$ and $\B{h}_i= \br{h_{ij}}_{j=1}^J$.

Under some regular conditions (see Theorem~\ref{them_asyproperty}), the asymptotic variance of $\widehat{m}(x)$ can be expressed as
\begin{equation}\label{eq_sigma}
    \Sigma(x, \B{\omega}, \B{\tau}, \B{h}) = a(x)\B{\omega}^{\top} \B{S}\br{\B{\tau}, \B{h}}\B{\omega}
\end{equation}
with
\begin{equation*}
    a(x) = \frac{\sigma^{2}\left(x\right)\int_{\R} K(v)^2 dv }{f_X\left(x\right)}, \quad
    \B{S}\br{\B{\tau}, \B{h}} = \operatorname{diag} \br{n_1^{-1}\BmR(\B{h}_1, \B{\tau}_1), \cdots, n_m^{-1}\BmR(\B{h}_m, \B{\tau}_m)},
\end{equation*}
\begin{equation}\label{eq_defJ}
    \BmR(\B{h}_i, \B{\tau}_{i}) = \left[\frac{\tau_{ij} \wedge \tau_{ij^{\prime}}-\tau_{ij} \tau_{ij^{\prime}}}{\sqrt{h_{ij} h_{ij^{\prime}}}f_{\varepsilon}\left(F_{\varepsilon}^{-1}(\tau_{ij})\right) f_{\varepsilon}\left(F_{\varepsilon}^{-1}(\tau_{ij^\prime})\right)}\right]_{j, j^\prime = 1}^J.
\end{equation}
Then the optimal weight vector that minimizes the asymptotic variance of $\widehat{m}(x)$ can be expressed as
\begin{equation}\label{eq1_optw_var}
    \B{\omega}^*\br{\B{\tau}, \B{h}} = \frac{c_1\B{d}_1 - c_2\B{d}_2}{c_1c_3 - c_2c_2},
\end{equation}
where
\begin{equation*}
    \B{d}_1 = \B{S}^{-1}\br{\B{\tau}, \B{h}} \B{1}_{mJ}, \quad \B{d}_2 = \B{S}^{-1}\br{\B{\tau}, \B{h}} \B{F}_{\varepsilon}^{-1}\br{\B{\tau}},
\end{equation*}
\begin{equation*}
    c_1 = \B{d}_2^{\top}\B{F}_{\varepsilon}^{-1}\br{\B{\tau}}, \quad c_2 = \B{d}_2^{\top}\B{1}_{mJ}, \quad c_3 = \B{d}_1^{\top} \B{1}_{mJ}
\end{equation*}
with $\B{F}_{\varepsilon}^{-1}\br{\B{\tau}} = \otimes_{i=1}^m \br{F_{\varepsilon}^{-1} \br{\tau_{ij}}}_{j=1}^J$.

\begin{remark}
    The optimal weights given by \eqref{eq1_optw_var} depend on the bandwidth vector $\B{h}$.
    Fortunately, $\B{h}$ can be empirically chosen as
    \begin{equation}\label{eq_bandwithform}
        \B{h} = \B{h}(\alpha, \nu) = \alpha  \otimes_{i=1}^m \br{n_i^{-\nu} \B{1}_J},
    \end{equation}
    where $\alpha$ and $\nu > 0$ are parameters.
    Then, the optimal weight vector in \eqref{eq1_optw_var} is independent of $\alpha$. Consequently, it can be written as
    \begin{equation}\label{eq2_optw_var}
        \B{\omega}^*\br{\B{\tau}, \B{h}(\alpha, \nu)} = \B{\omega}^*\br{\B{\tau}, \B{h}(1, \nu)}.
    \end{equation}
    Since it is often the case that the parameter $\nu$ can be determined empirically (e.g., for each local estimator, the theoretically optimal choice is $\nu=1/5$ ), the optimal weigh vector \eqref{eq2_optw_var} can be easily computed.

\end{remark}

\subsection{The choices of quantile levels}

The optimal weight vector in \eqref{eq1_optw_var} relies on the vector $\B{\tau}$ of quantile levels.
As an example, in this paper we only use the uniformly-spaced quantile levels in the form of
\begin{equation}\label{eq_taubar}
    \B{\tau} = \B{\tau}\br{\bar{\tau}} = \otimes_{i=1}^m \br{\bar{\tau}_{ij}}_{j=1}^J \text{ with }\bar{\tau}_{ij} = \bar{\tau} + \left(\frac{i+mj - m}{mJ} - \frac{1}{2} \br{1 + \frac{1}{mJ}}\right)d_{\tau},
\end{equation}
where $d_{\tau}$ controls the range of used quantile levels and $\bar{\tau}$ is the central position of all the quantile levels.
Naturally, $d_{\tau}$ and $\bar{\tau}$ should be constrained such that
\begin{equation}\label{eq_consttau}
    \br{\bar{\tau} - \frac{1}{2} d_{\tau}, \bar{\tau} + \frac{1}{2} d_{\tau}} \subset \br{\delta_{\tau}, 1 - \delta_{\tau}},
\end{equation}
where $\delta_{\tau}$ is an arbitrary small positive constant.
The choice of $\bar{\tau}$ is not unique in our approach.
In the simulation studies of this paper, a type of $\bar{\tau}$ is chosen as $\bar{\tau}^*$ satisfying
\begin{equation}\label{eq_bartau}
    \B{1}_{mJ}^{\top} \B{F}_{\varepsilon}^{-1}(\B{\tau}\br{\bar{\tau}^*}) = 0.
\end{equation}
With the choice of $\B{\tau} = \B{\tau}\br{\bar{\tau}^*}$, the uniform weight vector $\B{\omega}_u = 1/(mJ) \B{1}_{mJ}$ satisfies the constraints in \eqref{eq_nonbias}.
Thus, our optimal weight vector in \eqref{eq1_optw_var} performs at least better than $\B{\omega}_u$ under the criterion of asymptotic variance.
Moreover, when the distribution of random error is symmetric, we have $\bar{\tau}^* = 1/2$, then the resulting quantile levels $\B{\tau}\br{\bar{\tau}^*}$ is identical to the ones extensively used in those CQR-based works, e.g.,  \cite{Zou2008Composite, Kai2010, Sun2013Weighted, LinComposite2019}.

\begin{remark}
    Alternatively, we can also select $\bar{\tau} = \bar{\tau}^{**}$ with $\bar{\tau}^{**}$ satisfying
    \begin{equation}\label{eq2_bartau}
        \B{1}_{mJ}^{\top} \B{S}^{-1}\br{\B{\tau}\br{\bar{\tau}^{**}}, \B{h}\br{1, \nu}} \B{F}_{\varepsilon}^{-1}\br{\B{\tau}\br{\bar{\tau}^{**}}} = 0.
    \end{equation}
    Under $\B{\tau} = \B{\tau}\br{\bar{\tau}^{**}}$ and $\B{h} = \B{h}(\alpha, \nu)$, the optimal weight vector can be written as
    \begin{equation*}
        \B{\omega}^{**} = \B{\omega}^*\br{\B{\tau}\br{\bar{\tau}^{**}}, \B{h}(\alpha, \nu)} = \frac{\B{d}_3}{\B{1}_{mJ}^{\top}\B{d}_3} \ \text{ with } \ \B{d}_3 = \B{S}^{-1}\br{\B{\tau}\br{\bar{\tau}^{**}}, \B{h}\br{1, \nu}} \B{1}_{mJ}.
    \end{equation*}
    The above is a best choice because the resulting asymptotic variance reaches its minimum over all the weight constrained by $\B{\omega}^{\top}\B{1}_{mJ} = 1$ (instead of \eqref{eq_nonbias}), i.e.,
    \begin{equation}\label{eq_minwoptopt}
        \Sigma(x, \B{\omega}^{**}, \B{\tau}\br{\bar{\tau}^{**}}, \B{h}(\alpha, \nu)) = \min_{\B{\omega}^{\top}\B{1}_{mJ} = 1} \Sigma(x, \B{\omega}, \B{\tau}\br{\bar{\tau}^{**}}, \B{h}(\alpha, \nu)).
    \end{equation}

\end{remark}

%
%

\subsection{The sub-optimal bandwidths}\label{sec_opt_bw}

We only consider the case where the bandwidths are of the form of \eqref{eq_bandwithform} for simplicity.
By Theorem~\ref{thm_amse} and Corollary~\ref{coro_optamse} in the next section, to obtain the optimal convergence rate of AMSE$\br{\widehat{m}(x)}$, the optimal parameters $\nu$ and $\alpha$ should satisfy
\begin{equation}\label{eq_optnu}
    \nu^* = \frac{\ln n}{5\br{\ln n - \ln m}} \to \frac{1}{5s}, \quad \alpha^*(x, \B{\omega}) = \br{\frac{A_2^*\br{x, \B{\omega}}}{4 A_1^*\br{x, \B{\omega}}}}^{1/5},
\end{equation}
where
\begin{equation*}
    A_1^*\br{x, \B{\omega}} = \br{\sum_{i=1}^m\sum_{j=1}^J \omega_{ij} n_i^{-2\nu^*} \beta(x, \tau_{ij})}^2,
\end{equation*}
\begin{equation}\label{eq_A2}
    A_2^*\br{x, \B{\omega}} = a(x) \sum_{i=1}^m n_i^{\nu^*-1} \B{\omega}_i^{\top}\BmR_1\br{\B{\tau}_i}\B{\omega}_i \text{ with } \BmR_1\br{\B{\tau}_i} = \BmR(\B{1}_J, \B{\tau}_i).
\end{equation}
For constant bandwidths, the optimal $\alpha$ that minimizes $\text{AMISE}\br{\widehat{m}} = \int_{0}^1 \text{AMSE}\br{\widehat{m}(x)} W(x) d x$ can be expressed as
\begin{equation*}
    \alpha^*\br{\B{\omega}} = \br{\frac{\int_{\R}A_2^*\br{x, \B{\omega}}W(x) dx}{4 \int_{\R}A_1^*\br{x, \B{\omega}} W(x) dx}}^{1/5},
\end{equation*}
where $W(\cdot)$ is a given nonnegative weight function with a compact support in $[0,1]$.


If data batches $\D_i$ have the same size, i,e., $n_i=n/m$ for all $i$, we can choose that the bandwidth vector is of the form of
\begin{equation}\label{eq_bw_h1}
    \B{h} = h\B{1}_{mJ}, \mbox{ where } h = \alpha \br{\frac{n}{m}}^{-\nu}.
\end{equation}
Then the optimal bandwidth vector has a more concise form.
In fact, by Theorem~\ref{thm_amse} in the next section, and by minimizing the AMSE of $\widehat{m}(x)$, we get the optimal $h$ as
\begin{equation}\label{eq_bw_optch}
    h^*(x, \B{\omega}) = \br{\frac{a(x)}{\mu_2^2 m^{\prime\prime}(x)^2} \sum_{i=1}^m \frac{n}{n_i} \B{\omega}_i^{\top}\BmR_1\br{\B{\tau}_i}\B{\omega}_i}^{\frac{1}{5}} n^{-\frac{1}{5}},
\end{equation}
and the the optimal constant bandwidth as
\begin{equation}\label{eq_bw_optvh}
    h^*(\B{\omega}) = \br{\frac{\int_{\R} a(x) W(x) dx}{\mu_2^2 \int_{\R} m^{\prime\prime}(x)^2 W(x) dx} \sum_{i=1}^m \frac{n}{n_i} \B{\omega}_i^{\top}\BmR_1\br{\B{\tau}_i}\B{\omega}_i}^{\frac{1}{5}} n^{-\frac{1}{5}}.
\end{equation}

\begin{remark}
    Another natural idea to find the optimal weights, quantile levels, and bandwidths is to minimize the AMSE directly.
    However as shown in Corollary~\ref{thm_amse}, the function AMSE to be minimized is generally high dimensional and nonconvex with a complex structure,
    which means that the optimization problem is difficult to solve.
    Even worse, the minimizer is generally not unique, when the quantile level vector $\B{\tau}$ is considered as an independent parameter.
    Thus we consider a sub-optimal but more practical approach in this paper,
    which can drastically reduce the computational expense.
\end{remark}

As previously discussed, the choices of weights, quantile levels and bandwidths rely on the information of the CDF $F_{\varepsilon}(\cdot)$, PDF $f_{\varepsilon}(\cdot)$, the functions $\sigma(\cdot)$ and $\beta(\cdot, \cdot)$.
When these parameters and functions are unknown, they can be replaced by some consistent estimators.
More details will be discussed in Section~\ref{sec_praimple}.

\section{Theoretical Properties}\label{sec_theo}

\subsection{Regularity conditions}

To establish the theoretical properties of the global estimator $\widehat{m}(x)$, we introduce the following regularity conditions.
\begin{assumption}\label{assum_fFK}
    For some open interval $I_x \subseteq [0, 1]$, the functions $f_X(x)$ and $\sigma(x)$ are strictly positive for $x\in I_x$,
   their derivative functions $f_X^{\prime\prime}(x)$, $m^{\prime\prime}(x)$ and $\sigma^{\prime\prime}(x)$ exist and are continuous for $x\in I_x$. Moreover,
    the PDF $f_{\varepsilon}(u)$ is strictly positive and its derivative function $f_{\varepsilon}^{\prime}(u)$ is continuous and bounded for $u\in \R$.
\end{assumption}
\begin{assumption}\label{assum_K}
    The kernel function $K(u)$ is Lipschitz for $u\in \R$, nonnegative with a compact support and satisfies $\int_{\R} K(u) d u=1$, $\int_{\R} uK(u) d u = 0$, $\int_{\R} u^2K(u) d u < \infty$, $\int_{\R} K(u)^2 d u < \infty$.
    Moreover, there exist constants $\delta_K > 0$ and $s_K > 0$, such that $K(u) \geq \delta_K$ for $u \in [-s_K, s_K]$.
\end{assumption}
\begin{assumption}\label{assum_wtau}
    The vector $\B{\omega}$ of weights and vector $\B{\tau}$ of quantile levels satisfy \eqref{eq_nonbias}.
    Additionally, there exist constants $M_w$, $\delta_{\tau}$, $\lambda_{\min}$ and $\lambda_{\max}$ independent with $n$ and $m$, such that
    \begin{equation}\label{eq_Rtau}
        \|\B{\omega}\|_{\infty} < M_w < 0, \quad 0 < \delta_{\tau} < \tau_{ij} < 1 - \delta_{\tau}, \quad
    \end{equation}
    \begin{equation}\label{eq_lamR}
        0 < \lambda_{\min} <  \lambda_{\min}(\BmR_1\br{\B{\tau}_i}) \leq \lambda_{\max}(\BmR_1\br{\B{\tau}_i}) < \lambda_{\max} < +\infty
    \end{equation}
    for all $i = 1, \cdots, m$ and $j = 1, \cdots, J$,
    where $\BmR_1\br{\B{\tau}_i}$ is defined in \eqref{eq_A2}.
\end{assumption}
\begin{assumption}\label{assum_bat_bw}
    There exist constants $M_b > 0$ and $s \in (\frac{2}{3}, 1]$ such that
    \begin{equation}\label{eq_n_m}
        M_b^{-2} n^{s} < n_i < M_b^2 n^{s},  i = 1, \cdots, m.
    \end{equation}
    Moreover, the bandwidths satisfy
    \begin{equation}\label{eq_cr_h}
        h_{ij} = \Omega\br{n_i^{-\nu}} \text{ i.e. } h_{ij} = \Omega\br{n^{-s\nu}} \; \text{ with } \nu < 3 - \frac{2}{s} \; \text{ for } i = 1, \cdots, m,  j = 1, \cdots, J.
    \end{equation}
\end{assumption}


Assumptions~\ref{assum_fFK} and \ref{assum_K} are fairly general in nonparametric regression.
In Assumption~\ref{assum_wtau}, the constraint \eqref{eq_nonbias} is used to eliminate the non-negligible bias, and
additionally, the first condition in \eqref{eq_Rtau} requires a uniform boundedness of $\B{\omega}$, which is standard for designing weights, and the the second condition in \eqref{eq_Rtau} is necessary to obtain a uniformly convergence rate of the remainder term in \eqref{eq_asy_glomx}. Moreover, in Assumption~\ref{assum_wtau},
the condition in \eqref{eq_lamR} is standard in asymptotic analysis which guarantee the boundedness and nonsingularity of the covariance matrix in \eqref{eq_defJ}.
Assumption~\ref{assum_bat_bw} means that the sample size in each data batch is of the same order of $O(n/m)$, and the use of conditions on $s$ and $\nu$ is to guarantee that the remainder term in \eqref{eq_asy_glomx} is infinitesimal of higher order than the others. For the the assumptions above, we have the following further remarks.

\begin{remark}
It is worth pointing out that the condition on the error distribution in Assumption \ref{assum_fFK} is quite mild, which neither requires the error distribution to be symmetric as in \cite{Zou2008Composite, Kai2010} and \cite{Gu2020Sparse}, nor requires the error distribution to have a finite variance as in the local polynomial regression \citep{Fan1993Local}.
In the case of massive data, it is usually quite difficult to obtain the exact information about the error distribution. Hence, our method is desirable for its general applicability in term of error distribution. \end{remark}

\begin{remark}

The condition \eqref{eq_n_m} in Assumption~\ref{assum_bat_bw} implies that the number $m$ of data batches is restricted, namely, $m = O\br{n^{\gamma}} $ with $0\leq\gamma < \frac{1}{3}$.
This restriction is necessary to suppress the remainder term $R(\B{\omega}, \B{h})$ and to establish the asymptotic normality in \eqref{lim_sqrtn_nonp}.
However, this is a common condition in DC quantile methods with one-shot aggregation, even for parametric models, this condition is commonly used \cite[see e.g.][]{zhao2015general, Rong2018Composite, Chen2020Quantile}.
\end{remark}


\subsection{Asymptotic properties}

For simplicity, we define
\begin{equation*}
    Q_{Y, h}^*(\tau|x) = a^*  \ \text{ with }\ \br{a^*, b^*} = \argmin_{(a, b)\in \R\times \R} \E{\rho_{\tau}\left(Y - a - b(X - x)\right) K_{h}\left(X - x\right)},
\end{equation*}
and
\begin{equation}\label{eq_phi}
    \phi(x, \tau, h, \D_i) = \frac{1}{n_i h v\left(x; \tau, h\right) } \sum_{j=1}^{n_i} \psi_{\tau}\left( Y_{ij}- Q_{Y, h}^*(\tau|X_{ij}) \right) K\left(\frac{X_{ij}-x}{h}\right)\ \text{ for } 1 \leq i \leq m
\end{equation}
with
\begin{equation*}
    v\left(x; \tau, h\right) = \int_{\R} \frac{K(u)}{\sigma(x + h u)} f_{X}\left(x+h u \right) f_{\varepsilon}\left(\frac{Q_{Y, h}^*(\tau|x) - m(x + h u)}{\sigma(x + h u)}\right) d u.
\end{equation*}

The following theorem states the asymptotic properties of the global estimator $\widehat{m}(x)$.

\begin{theorem}\label{them_asyproperty}
    Under Assumptions \ref{assum_fFK}-\ref{assum_bat_bw}, the global estimator $\widehat{m}(x)$ defined by \eqref{nonp_glo_esti} has the following asymptotic representation:
    \begin{equation}\label{eq_them_asyproperty}
        \widehat{m}(x) = m(x) + B(x, \B{\omega}, \B{\tau}, \B{h}) + \sum_{i=1}^m\sum_{j=1}^J \omega_{ij}\phi(\tau_{ij}, h_{ij}, \D_i) + R(\B{\omega}, \B{h})\ \mbox{ for } x \in [0, 1],
    \end{equation}
 where
    \begin{equation}\label{eq_defB}
        B(x, \B{\omega}, \B{\tau}, \B{h}) = \sum_{i=1}^m\sum_{j=1}^J \omega_{ij} h_{ij}^2 \beta(x, \tau_{ij})
    \end{equation}
    and
    \begin{equation*}
        R(\B{\omega}, \B{h}) = O_p\br{\|\B{\omega}\|_2 \sqrt{m}\br{\log n}^{\frac{3}{4}} n^{-\frac{3}{4} s \br{1-\nu}}} + o\br{\|\B{\omega}\|_1 n^{-2s\nu}}.
    \end{equation*}
    Consequently, we have
    \begin{equation}\label{lim_sqrtn_nonp}
        \frac{\widehat{m}(x) - m(x) - B(x, \B{\omega}, \B{\tau}, \B{h})}{\sqrt{\Sigma(x, \B{\omega}, \B{\tau}, \B{h})}} \overset{d}{\to} N(0, 1) \ \mbox{ for } x \in [0, 1],
    \end{equation}
    where ``$\overset{d}{\to}$" stands for the convergence in distribution, and $\beta(x, \tau_{ij})$ and $\Sigma(x, \B{\omega}, \B{\tau}, \B{h})$ are defined in \eqref{eq_defbeta} and \eqref{eq_sigma}, respectively.

\end{theorem}

In Theorem~\ref{them_asyproperty}, the asymptotic distribution of $\widehat{m}(x)$ can throw light on the reasonability of the choices of the optimal weights, bandwidths and quantile levels proposed in the previous section.
For the theorem, we have the following remark.
\begin{remark}

    When the bandwidths have the form $\B{h} = h\B{1}_{mJ}$, by the theorem together with the conditions in~\eqref{eq_nonbias}, we can write the bias term \eqref{eq_defB} in the following concise form
    \begin{equation}\label{eq_defB2}
        B(x, \B{\omega}, \B{\tau}, \B{h}) = \frac{\mu_2}{2}h^2 m^{\prime\prime}(x)\ \mbox{ for } x \in [0, 1],
    \end{equation}
    which is no longer dependent on $\B{\omega}$ and $\B{\tau}$. From Theorem~\ref{them_asyproperty}, we can see that  despite the local estimations on each data set have a non-ignorable bias, the bias of our global estimator is even of the order of $O(\left\|\B{h}\right\|_{\infty}^2)$, which is identical to the standard convergence rate of kernel regressions obtained by using the full data set.
This is owing to our strategy of multilevel quantile-matched composite, in which the weights and quantile levels are well-chosen, such that the non-vanishing biases in local estimators are globally corrected.

\end{remark}

\subsection{Asymptotic mean square error}

If the remainder term $R(\B{\omega}, \B{h})$ in the asymptotic representation (\ref{eq_them_asyproperty}) is ignored, we have the following theorem,  which clarifies the asymptotic mean square error of our estimator.

\begin{theorem}\label{thm_amse}

    Under the conditions of Theorem~\ref{them_asyproperty}, when the vector of bandwidths are of the form of \eqref{eq_bandwithform}, the AMSE of $\widehat{m}(x)$ can be expressed as
    \begin{equation}\label{eq1_amse}
        \text{AMSE}\br{\widehat{m}(x)} = \alpha^4 \br{\sum_{i=1}^m\sum_{j=1}^J \omega_{ij} n_i^{-2\nu} \beta(x, \tau_{ij})}^2 + \frac{a(x)}{\alpha}\sum_{i=1}^m \frac{1}{n_i^{1-\nu}} \B{\omega}_i^{\top}\BmR_1\br{\B{\tau}_i}\B{\omega}_i  \mbox{ for } \ x \in [0, 1].
    \end{equation}
    When the vector of bandwidths further satisfies the form $\B{h} = h\B{1}_{mJ}$ with the scaler $h = \Omega(n^{-s\nu})$, the AMSE of $\widehat{m}(x)$ can be simplified to
    \begin{equation}\label{eq1d1_amse}
        \text{AMSE} \br{\widehat{m}(x)} = \frac{\mu_2^2}{4} m^{\prime\prime}(x)^2 h^4 + \frac{a(x)}{h}\sum_{i=1}^m \frac{1}{n_i} \B{\omega}_i^{\top}\BmR_1\br{\B{\tau}_i}\B{\omega}_i \ \mbox{ for } x \in [0, 1].
    \end{equation}
    In either case, it holds that
    \begin{equation}\label{eq2_amse}
        \text{AMSE}\br{\widehat{m}(x)} = O\br{\|\B{\omega}\|_1^2 n^{-4s\nu}} + \Omega\br{\|\B{\omega}\|_2^2 n^{-s\br{1-\nu}}}\  \mbox{ for } x \in [0, 1].
    \end{equation}

\end{theorem}

Theorem~\ref{thm_amse} shows that the AMSE of our estimator depends on the used weights, quantile levels and bandwidths.
Moreover, as shown in \eqref{eq2_amse}, the convergence rate of $\widehat{m}(x)$ not only relies on the parameter $\nu$ but also relies on the weights.
Before finding the optimal convergence rate of our estimator, it is necessary to discuss the properties of the optimal weights in \eqref{eq1_optw_var}.
The following theorem justifies the optimality of the weights given by \eqref{eq1_optw_var}.

\begin{theorem}\label{thm_optw}
    Under the conditions of Theorem~\ref{them_asyproperty}, the weight vector in \eqref{eq1_optw_var} minimizes the asymptotic variance of $\widehat{m}(x)$.
    Moreover, under the quantile levels $\B{\tau} = \B{\tau}\br{\bar{\tau}^*}$ or $\B{\tau} =  \B{\tau}\br{\bar{\tau}^{**}}$ which satisfy \eqref{eq_bartau} and \eqref{eq2_bartau}, respectively, it holds that
    \begin{equation}\label{eq_minvar_wu}
        \Sigma\br{x, \B{\omega}^*, \B{\tau}, \B{h}} \leq \Sigma\br{x, \B{\omega}_u, \B{\tau}, \B{h}} \text{ with } \B{\omega}_u = \frac{1}{mJ} \B{1}_{mJ}.
    \end{equation}
\end{theorem}

Theorem~\ref{thm_optw} shows that if the quantile levels satisfy \eqref{eq_bartau} or \eqref{eq2_bartau}, the use of the optimal weight vector can achieve a smaller asymptotic variance of $\widehat{m}(x)$ than that obtained by the use of the uniform weight vector.
This conclusion will be useful for the efficiency of our estimator discussed later. For the choice of weight vector, we have the following result.


\begin{corollary}\label{coro_optw_norm}
    Under the conditions of Theorem~\ref{them_asyproperty}, let the quantile levels $\B{\tau}\br{\bar{\tau}^*}$ and $\B{\tau}\br{\bar{\tau}^{**}}$ satisfy \eqref{eq_bartau} and \eqref{eq2_bartau}, respectively.
    Then for $\B{\omega} = \B{\omega}^*\br{\B{\tau}\br{\bar{\tau}^*}, \B{h}}$ and $\B{\omega} = \B{\omega}^*\br{\B{\tau}\br{\bar{\tau}^{**}}, \B{h}}$, the following holds:
    \begin{equation}\label{eq_optw_norm}
        \|\B{\omega}\|_1 = O(1), \quad \|\B{\omega}\|_2^2 = \Omega(m^{-1}).
    \end{equation}

\end{corollary}

The above indicates the regularity of the selected weights.
On the other hand, combining Theorem~\ref{thm_amse} and Corollary~\ref{coro_optw_norm}, we immediately obtain the convergence rate of the estimator $\widehat{m}(x)$ with the optimal weights.

\begin{corollary}\label{coro_optamse}

    Under the conditions of Theorem~\ref{them_asyproperty},
    given the bandwidths $\B{h} = \B{h}\br{\alpha, \nu}$ and the weights
    \begin{equation}
        \B{\omega} = \B{\omega}^*\br{\B{\tau}\br{\bar{\tau}^*}, \B{h}\br{\alpha, \nu}} \quad \text{ or } \quad  \B{\omega} = \B{\omega}^*\br{\B{\tau}\br{\bar{\tau}^{**}}, \B{h}\br{\alpha, \nu}},
    \end{equation}
    then, the following conclusions hold:
    \begin{itemize}
        \item If $11/15 \leq s \leq 1$, we have $\text{AMSE}\br{\widehat{m}(x)} = O\br{n^{-4s\nu}+n^{s\nu-1}}$ for $x \in [0, 1]$. If $2/3 < s \leq 11/15$, we have $\text{AMSE}\br{\widehat{m}(x)} = O\br{n^{-4s\nu}}$ for $x \in [0, 1]$.
         \item Particularly, when $11/15 \leq s \leq 1$ and $\nu = \nu^*$ given in \eqref{eq_optnu}, the global estimator $\widehat{m}(x)$ enjoys the optimal convergence rate as
        \begin{equation*}
            \text{AMSE}\br{\widehat{m}\br{x}} = O(n^{-4/5})\  \mbox{ for } x \in [0, 1].
        \end{equation*}
    \end{itemize}
\end{corollary}

The above shows the oracle convergence rate of our estimator, which is the same as that of the entire data estimator.

\subsection{Asymptotic relative efficiency}

In this subsection, we establish the asymptotic relative efficiency (ARE) between our estimator $\widehat{m}\br{x}$ and the oracle local linear estimator $\widehat{m}_{\text{oll}}(x)$ obtained on the entire data set $\D$ by comparing their AMSEs.
Throughout this subsection, we suppose Assumptions~\ref{assum_fFK}-\ref{assum_bat_bw} hold with $11/15 \leq s \leq 1$ and the vector of bandwidth satisfies $\B{h} = h\B{1}_{mJ}$.
Then the AMSE of our estimator is given by \eqref{eq1d1_amse}.
Minimizing $\text{AMSE}\br{\widehat{m}(x)}$ and $\text{AMISE}\br{\widehat{m}}$, the optimal variable bandwidth and constant bandwidth for our estimator are given by \eqref{eq_bw_optch} and \eqref{eq_bw_optvh}, respectively. On the other hand, according to \cite{Fan1993Local, Kai2010}, under some regular conditions, the oracle local linear estimator $\widehat{m}_{\text{oll}}(x)$ for $m(x)$ has AMSE as
\begin{equation}\label{eq_amse_lls}
    \text{AMSE}\br{\widehat{m}_{\text{oll}}(x)} = \frac{\mu_2^2}{4} m^{\prime\prime}(x)^2 h^4 + \frac{a(x)}{nh} \  \mbox{ for } x \in [0, 1].
\end{equation}
Minimizing $\text{AMSE}\br{\widehat{m}_{\text{oll}}(x)}$ and $\text{AMISE}\br{\widehat{m}_{\text{oll}}}$, respectively, the optimal variable bandwidth and constant bandwidth for $\widehat{m}_{\text{oll}}(x)$ can be respectively expressed as
\begin{equation}\label{eq_bw_opth_lls}
    h_{\text{oll}}^*\br{x} = \br{\frac{a(x)}{\mu_2^2 m^{\prime\prime}(x)^2}}^{\frac{1}{5}} n^{-\frac{1}{5}}, \quad h_{\text{oll}}^* = \br{\frac{\int_{\R} a(x) W(x) dx}{\mu_2^2 \int_{\R} m^{\prime\prime}(x)^2 W(x) dx}}^{\frac{1}{5}} n^{-\frac{1}{5}}.
\end{equation}

The ARE between $\widehat{m}\br{x}$ and $\widehat{m}_{\text{oll}}(x)$ is defined by
\begin{equation*}
    \text{ARE}\br{\widehat{m}\br{x}, \widehat{m}_{\text{oll}}(x)} = \frac{\text{AMSE}_{\text{opt}}\br{\widehat{m}_{\text{oll}}(x)}}{\text{AMSE}_{\text{opt}}\br{\widehat{m}\br{x}}},
\end{equation*}
where $\text{AMSE}_{\text{opt}}\br{\widehat{m}\br{x}}$ and $\text{AMSE}_{\text{opt}}\br{\widehat{m}_{\text{oll}}(x)}$ respectively denote the AMSE of $\widehat{m}\br{x}$ and $\widehat{m}_{\text{oll}}(x)$ evaluated at the optimal variable bandwidths given in \eqref{eq_bw_optch} and \eqref{eq_bw_opth_lls}.
Then by straightforward calculations, it follows from the AMSEs in \eqref{eq1d1_amse} and \eqref{eq_amse_lls} that
\begin{equation}\label{eq_are}
    \text{ARE}\br{\widehat{m}\br{x}, \widehat{m}_{\text{oll}}(x)} = \br{\sum_{i=1}^m \frac{n}{n_i} \B{\omega}_i^{\top}\BmR_1\br{\B{\tau}_i}\B{\omega}_i}^{-\frac{4}{5}}.
\end{equation}
From \eqref{eq_are}, we can see that the ARE not only depends on the error distribution, but also depends on the used weights and quantile levels.
Moreover, the relevance among them is rather complex.
To address this issue, we turn to discuss the lower bound of the ARE.

Apart from Assumptions~\ref{assum_fFK}-\ref{assum_bat_bw} mentioned before, we further assume that
the data batches $\D_i$ have the same size, i,e., $n_i=n/m$ for all $i$,
    the weights are given by $\B{\omega} = \B{\omega}^*\br{\B{\tau}, \B{h}(1, \nu^*)}$ with the quantile levels $\B{\tau} = \B{\tau}\br{\bar{\tau}^*}$ satisfying \eqref{eq_bartau}.
Then, Theorem~\ref{thm_optw} leads to a lower bound for ARE as
\begin{equation}\label{eq_areRu}
    \text{ARE}\br{\widehat{m}\br{x}, \widehat{m}_{\text{oll}}(x)} \geq R_{\text{u}} = \br{\frac{1}{m}\sum_{i=1}^m \frac{1}{J^2} \B{1}_{J}^{\top}\BmR_1\br{\B{\tau}_i}\B{1}_J}^{-\frac{4}{5}}.
\end{equation}
In \eqref{eq_areRu}, $R_{\text{u}}$ is just the ARE between the local CQR estimator with uniform weights and the oracle local linear estimator $\widehat{m}_{\text{oll}}(x)$.
According to the values of $J$, we have the following observations:
\begin{itemize}
    \item For relatively small $J$, it was shown by Table 3.1 in \cite{Kai2010} that
    under some typical non-normal error distributions, e.g., Laplace distribution, $t$-distribution and mixtures of normal distributions, it holds that $R_{\text{u}} > 1$.
    This implies that our estimator is able to improve the estimation efficiency for relatively small $J$ if the error distribution is non-normal.
    \item For large $J$, Theorem 2 in \cite{Kai2010} shows that regardless of whether the error is normal or not, $R_{\text{u}} \to 1$ as $J \to \infty$ and $d_{\tau} \to 1$.
    Recalling \eqref{eq_bartau}, we have $\bar{\tau}^{*} = \frac{1}{2}$ for symmetric error distributions, which implies that $d_{\tau}$ can be close to $1$ and the constraint in \eqref{eq_consttau} can be satisfied.
    Thus, for symmetric error distributions including the normal distribution, our estimator does not lose any efficiency, as long as $J$ is large and $d_{\tau}$ approaches $1$.
\end{itemize}


\section{Practical Implementation}\label{sec_praimple}


\subsection{Pilot estimations for weights and quantile levels}\label{sec_pil_est}

As discussed in Section~\ref{sec_meth}, the choices of weights and quantile levels rely on the information of CDF $F_{\varepsilon}(\cdot)$, PDF $f_{\varepsilon}(\cdot)$ and $\sigma(\cdot)$.
When these functions are unknown, they can be replaced by the corresponding estimators.
As an example, we propose the following pilot estimation procedures.

The functions $F_{\varepsilon}(\cdot)$, $f_{\varepsilon}(\cdot)$ and $\sigma(\cdot)$ can be estimated as follows:
\begin{enumerate}
    \item Construct the N-W estimator of the DC form as an initial estimator for $m(\cdot)$ by
        \begin{equation}\label{eq_pilmnw}
            \widehat{m}_{\text{nw}}(x) = \frac{\sum_{i=1}^m S_i^e(x)}{\sum_{i=1}^m T_i(x)}, \quad x \in [0,1],
        \end{equation}
        where $S_i^e(x)$ and $T_i(x)$ are computed on the $i$-th data batch by
        \begin{equation*}
            S_i^e(x) = \sum_{j=1}^{n_i} Y_{ij}K_h(X_{ij} - x), \quad T_i(x) = \sum_{j=1}^{n_i} K_h(X_{ij} - x).
        \end{equation*}
    \item Estimate $\sigma(\cdot)$ by N-W estimator of the DC form by
        \begin{equation}\label{eq_pilsgmnw}
            \widehat{\sigma}^2(x) = \frac{\sum_{i=1}^m S_i^{\sigma}(x)}{\sum_{i=1}^m T_i(x)}, \quad x \in [0,1],
        \end{equation}
        where $S_i^{\sigma}(x) = \sum_{j=1}^{n_i} \left(Y_{ij} - \widehat{m}_{nw}(X_{ij}) \right)^2K_h(X_{ij} - x)$ is computed on the $i$-th data batch.
    \item Let $\varepsilon_{ij} = \frac{Y_{ij} - \widehat{m}_{\text{nw}}(X_{ij})}{\widehat{\sigma}(X_{ij})}$, for $j = 1, \cdots, n_i$ and $i = 1, \cdots, m$. For $\tau \in (0, 1)$ and $x \in [0,1]$, estimate $f_{\varepsilon}\left(x\right)$ and $F_{\varepsilon}(x)$ by
        \begin{equation*}
            \widehat{f}_{\varepsilon}\left(x\right) = \frac{1}{m} \sum_{i=1}^m\widehat{f}_{\varepsilon, i}\left(x\right), \quad \widehat{F}_{\varepsilon}(x) = \int_{-\infty}^{x} \widehat{f}_{\varepsilon}\left(x\right) dx,
        \end{equation*}
        respectively, where $\widehat{f}_{\varepsilon}\left(x\right)$ is local estimator computed on the $i$-the data batch by
        \begin{equation*}
            \widehat{f}_{\varepsilon, i}\left(x\right) = \frac{1}{n_i} \sum_{j=1}^{n_i} K_h(\widehat{\varepsilon}_{ij} - x).
        \end{equation*}
\end{enumerate}

\subsection{Practical bandwidth selection}\label{sec_bw_select}

The bandwidth selection is an important issues in practical application.
According to the idea in \cite{Kai2010}, we give two bandwidth selectors as follows:
\begin{itemize}
    \item The ``pilot'' selector: We can use the optimal bandwidth formulas in Subsection~\ref{sec_opt_bw} to estimate the optimal
    bandwidth.
    The unknown functions $F_{\varepsilon}$, $f_{\varepsilon}$ and $\sigma$ can be replaced with the associated pilot estimators given in the previous subsection, and $\beta(x, \tau_{ij})$ can be estimated by
    \begin{equation*}
        \widehat{\beta}(x, \tau_{ij}) = \mu_2 \widehat{b}_{ij}^{(2)}, \quad i = 1, \cdots, m, \quad j = 1, \cdots, J,
    \end{equation*}
    where $\widehat{b}_{ij}^{(2)}$ is the local cubic CQR estimators defined by
    \begin{equation*}
        \br{\widehat{a}_i, \widehat{b}_{ij}^{(1)}, \widehat{b}_{ij}^{(2)}} = \argmin_{\br{a, b^{(1)}, b^{(2)}} \in \R^3} \sum_{l=1}^{n_i} \rho_{\tau_{ij}}\left(Y_{il} - a - \sum_{k=1}^2 b^{(k)} (X_{il} - x)^k\right) K_{h_p}\left(X_{il} - x\right),
    \end{equation*}
    with $h_p$ being a pre-given pilot bandwidth.

    \item A short-cut strategy: Recalling the optimal bandwidths in \eqref{eq_bw_optch} and \eqref{eq_bw_optvh}, and comparing them with the ones in \eqref{eq_bw_opth_lls}, we can obtain the following relationships:
    \begin{equation*}
        h^*(x, \B{\omega}) = V^{\frac{1}{5}}\br{\B{\omega}, \B{\tau}} h_{\text{oll}}^*\br{x}, \   h^*(\B{\omega}) = V^{\frac{1}{5}}\br{\B{\omega}, \B{\tau}} h_{\text{oll}}^*\ \text{ with }  V\br{\B{\omega}, \B{\tau}} = \sum_{i=1}^m \frac{n}{n_i} \B{\omega}_i^{\top}\BmR_1\br{\B{\tau}_i}\B{\omega}_i,
    \end{equation*}
    which shows that the optimal bandwidths for $\widehat{m}(x)$ can be estimated via the optimal bandwidths for $\widehat{m}_{\text{oll}}(x)$.
    Actually, the optimal bandwidths for the local least squares estimators can be chosen by the existing bandwidth selectors, and $V\br{\B{\omega}, \B{\tau}}$ can be replaced with the associated plug-in estimator by using the pilot estimators in the previous subsection.

\end{itemize}

\subsection{The entire estimation procedure}

To be clearer, we summary the entire estimation procedure for $m(x)$ as follows:
\begin{enumerate}
    \item Carry out the pilot estimations in Section~\ref{sec_pil_est} to obtain the estimators for $F_{\varepsilon}(\cdot)$, $f_{\varepsilon}(\cdot)$ and $\sigma(\cdot)$.
    \item After replacing the unknown parameters and functions with the corresponding estimators, compute the quantile level vector $\B{\tau}* = \B{\tau}\br{\bar{\tau}^*}$ with $\bar{\tau}^*$ solved from \eqref{eq_bartau}, and then compute the optimal weight vector $\B{\omega}^* = \B{\omega}^*\br{\B{\tau}^*, \B{h}\br{1, \nu^*}}$ by the formula \eqref{eq1_optw_var} with $\nu^*$ given by \eqref{eq_optnu}.
    \item Select the optimal bandwidth vector $\B{h}^* = \B{h}\br{\alpha^*, \nu^*}$ based on the selectors in Section~\ref{sec_bw_select}.
    \item For $i = 1, \cdots, m$, compute the local estimator $\widehat{m}_{ij}^*(x) = \widehat{m}_i \br{x; \tau_{ij}^*, h_{ij}^*}$ by \eqref{eq_hatm} on the $i$-th data batch.
    \item Compute the optimal global estimator for $m(x)$ by
    \begin{equation*}
        \widehat{m}^*(x) = \sum_{i=1}^m \omega_{ij}^* \widehat{m}_{ij}^* (x).
    \end{equation*}

\end{enumerate}

\section{Numerical analyses}\label{sec_simu}

In this section, we conduct simulation studies and real data analyses to verify performance of the optimal global estimator $\widehat{m}^*(\cdot)$.
We consider two competing estimators.
The first one is a benchmark estimator, namely, the oracle local linear estimator denoted by $\widehat{m}_{\text{oll}}\br{\cdot}$, which is computed by the full data set with the bandwidths selected via cross validation.
The second competitor is the simple-average local linear least absolute deviation (ALAD) estimator denoted by $\widehat{m}_{\text{alad}}\br{\cdot}$, which is actually the simple average of the local linear least absolute deviation (LAD) estimators from each data batch, in other words, it can be expressed by
\begin{equation}
    \widehat{m}_{\text{alad}}\br{x} = \frac{1}{m} \sum_{i=1}^m \widehat{m}_i\br{x; 0.5, h}, \quad x \in [0, 1],
\end{equation}
where the bandwidth $h$ is obtained by the short-cut strategy introduced in Section~\ref{sec_bw_select}.
In the procedure of simulation, we use the Epanechnikov kernel as $K(z)=\frac{3}{4}\left(1-z^{2}\right)_{+}$.

\subsection{Simulation Studies}

In this subsection, we will carry out simulation studies to make a comprehensive comparison between our estimators and the two competing estimators.
We consider various experiment conditions, such as symmetric or asymmetric errors, and homoscedastic or heteroscedastic models.
We also consider the mixtures of two distributions, more specifically, a mixture distribution can be chosen as $(1-\lambda)F_{\varepsilon}+\lambda F_{\sqrt{10}\,\varepsilon}$ with mixture proportion denoted by $\lambda$, where $F_{\varepsilon}$ is distribution function of $\varepsilon$.
If without special statement, all the distributions of $\varepsilon$ involved in the simulations are centralized.

The total number $n$ of samples in the full data set is taken as $10000$.
The full data set is divided equally into $m$ data batches.
To show the influence of the number of data batches, we successively take $m = 1, 5, 10, 20, 50$.
For constructing the composite estimator $\widehat{m}^*(\cdot)$, the quantile levels are chosen as the form of \eqref{eq_taubar} with $J=5$ and $d_{\tau} = 0.5$,
the bandwidths are obtained by the short-cut strategy introduced in Section~\ref{sec_bw_select}.
The number of replications in the simulation is designed as $400$.
The performance of any estimator $\widehat{g}(\cdot)$ of a function $g(\cdot)$ is evaluated by the average squared errors (ASEs) defined by
\begin{equation*}
    \operatorname{ASE}(\widehat{g})=\frac{1}{n_{\text{grid}}} \sum_{i=1}^{n_{\text {grid }}}\left|\widehat{g}\left(u_{k}\right)-g\left(u_{k}\right)\right|^{2},
\end{equation*}
where $\left\{u_{k}, k=1, \ldots, n_{\text {grid }}\right\}$ are the grid points at which the estimator $\widehat{g}(\cdot)$ is evaluated.
In this subsection, we take $n_{\text {grid }} = 200$ and let the grid points be evenly distributed over the interval on which $\widehat{g}(\cdot)$ is estimated.
We use the ratio of average squared errors (RASEs)
\begin{equation*}
    \text{RASE}\br{\widehat{g}_1,\widehat{g}_2} = \frac{\operatorname{ASE}(\widehat{g}_2)}{\operatorname{ASE}
    (\widehat{g}_1)}
\end{equation*}
to compare the performance of two estimators $\widehat{g}_1$ and $\widehat{g}_2$.

\subsubsection{Homoscedastic model}

We first consider the homoscedastic model from \cite{Fan1992Variable}:
\begin{equation}\label{homodel_nonp}
    Y=m(X)+\sigma \varepsilon,
\end{equation}
where the covariate $X$ follows $N(0, 1)$, the regression function is chosen as $m(x) = \sin (2 x)+2 \exp \left(-16 x^{2}\right)$ and the standard deviation is $\sigma=0.5$.
By taking various error distributions, we manage to estimate $m(x)$ for $x\in [-1.5, 1.5]$.

\begin{itemize}
    \item Example 1a. We consider the model \eqref{homodel_nonp} with various symmetric error distributions. The means and standard derivations of RASEs are reported in Table~\ref{tab_homodel_nonp_sym}.
    \item Example 1b. We consider the model \eqref{homodel_nonp} with various asymmetric error distributions. The means and standard derivations of RASEs are reported in Table~\ref{tab_homodel_nonp_asym1}.
\end{itemize}

From the values of RASE$(\cdot, \cdot)$ in Tables~\ref{tab_homodel_nonp_sym} - \ref{tab_homodel_nonp_asym1}, we have the following findings:

(i) When the number $m$ of data batches is no more than $20$, the values of RASE$(\widehat{m}^*, \widehat{m}_{\text{oll}})$ are significantly greater than $1$ in most cases,
indicating that $\widehat{m}^*(\cdot)$ outperforms $\widehat{m}_{\text{oll}}\br{\cdot}$ uniformly for all choices of the mixture proportion $\lambda$. Generally, the values of RASE$(\widehat{m}^*, \widehat{m}_{\text{oll}})$ increase with $\lambda$ increasing.
This means that our estimator is more robust than the local linear estimator.
However, sometimes the above advantage vanishes gradually as $m$ increases, and particularly, when $m = 50$, the advantage of $\widehat{m}^*(\cdot)$ depends on the choices of the mixture proportion $\lambda$. This illustrates that in the asymptotic expansion of \eqref{eq_asy_glomx}, the effect of the remainder term $\mathbb{I}_4$ become more significant for larger $m$.

(ii)  Our estimator $\widehat{m}^*(\cdot)$ performs significantly better than $\widehat{m}_{\text{alad}}\br{\cdot}$ in most cases.
Unlike the observation in (1), however, the values of RASE$(\widehat{m}^*, \widehat{m}_{\text{alad}})$ seems to be insensitive to $m$ and $\lambda$. It is because
both estimators $\widehat{m}^*(\cdot)$ and $\widehat{m}_{\text{alad}}\br{\cdot}$ are quantile-based methods, implying that they have the same sensitivity to $m$, and the same robustness against outliers.

(iii) All the values of RASE$(\widehat{m}_{\text{alad}}, \widehat{m}_{\text{oll}})$ under asymmetric distributions are close to $0$, indicating that the ALAD estimator can be inconsistent when the distribution of $\varepsilon$ is asymmetric.
Contrarily, our estimator $\widehat{m}^*(\cdot)$ performances  satisfactorily in this case, which further demonstrates the effectiveness of the strategy of multilevel quantile-matched composite.

\begin{table}[t]
    \centering
    \caption{The means and standard derivations of RASEs in Example 1a}\label{tab_homodel_nonp_sym}
    \renewcommand\arraystretch{0.9}
    \resizebox{\textwidth}{!}{%
    \begin{tabular}{@{}ccc|ccccc|ccccc@{}}
        \toprule
        \multirow{2}{*}{Distributions} & \multirow{2}{*}{$\lambda$} &  & \multicolumn{5}{c|}{RASE$(\widehat{m}^*, \widehat{m}_{\text{oll}})$} & \multicolumn{5}{c}{RASE$(\widehat{m}^*, \widehat{m}_{\text{alad}})$} \\ \cmidrule(l){4-13}
         &  &  & $m = 1$ & $m = 5$ & $m = 10$ & $m = 20$ & $m = 50$ & $m = 1$ & $m = 5$ & $m = 10$ & $m = 20$ & $m = 50$ \\ \midrule
        \multirow{8}{*}{N$(0, 1)$} & \multirow{2}{*}{0} & mean & 1.0254 & 1.0305 & 0.9616 & 0.9118 & 0.7438 & 1.2079 & 1.3462 & 1.1612 & 1.1660 & 1.2533 \\
         &  & std & 0.1252 & 0.1341 & 0.1616 & 0.1264 & 0.1736 & 0.2140 & 0.2014 & 0.2256 & 0.2069 & 0.3747 \\ \cmidrule(l){2-13}
         & \multirow{2}{*}{0.05} & mean & 1.3159 & 1.2548 & 1.2184 & 1.1479 & 1.0005 & 1.2162 & 1.2379 & 1.1495 & 1.0802 & 1.5681 \\
         &  & std & 0.1833 & 0.1773 & 0.1438 & 0.1881 & 0.2402 & 0.2400 & 0.2584 & 0.1791 & 0.2025 & 0.4424 \\ \cmidrule(l){2-13}
         & \multirow{2}{*}{0.1} & mean & 1.5727 & 1.5862 & 1.6089 & 1.5197 & 1.1865 & 1.1187 & 1.2572 & 1.3113 & 1.3332 & 1.1835 \\
         &  & std & 0.2295 & 0.2448 & 0.2524 & 0.2467 & 0.2859 & 0.2224 & 0.2276 & 0.2720 & 0.2469 & 0.2838 \\ \cmidrule(l){2-13}
         & \multirow{2}{*}{0.2} & mean & 1.8212 & 2.2943 & 2.1070 & 2.2799 & 1.6583 & 0.8811 & 1.2282 & 1.2099 & 1.4501 & 1.3078 \\
         &  & std & 0.2275 & 0.3266 & 0.3192 & 0.3385 & 0.4646 & 0.1361 & 0.2588 & 0.2325 & 0.2569 & 0.3469 \\ \midrule
        \multirow{8}{*}{Standard Laplace} & \multirow{2}{*}{0} & mean & 2.9770 & 2.6554 & 2.4556 & 1.9365 & 1.5956 & 1.0118 & 1.3706 & 1.3092 & 1.2974 & 1.3106 \\
         &  & std & 0.4688 & 0.4135 & 0.4934 & 0.2740 & 0.2676 & 0.1631 & 0.1983 & 0.2218 & 0.2348 & 0.2218 \\ \cmidrule(l){2-13}
         & \multirow{2}{*}{0.05} & mean & 3.3970 & 3.4080 & 3.3435 & 2.8602 & 2.0082 & 0.7931 & 1.2049 & 1.4851 & 1.1493 & 1.1426 \\
         &  & std & 0.4085 & 0.6973 & 0.4937 & 0.4882 & 0.5291 & 0.1147 & 0.2227 & 0.2249 & 0.1954 & 0.2605 \\ \cmidrule(l){2-13}
         & \multirow{2}{*}{0.1} & mean & 4.7786 & 4.7811 & 4.4703 & 3.7875 & 3.0340 & 0.8356 & 1.5946 & 1.3792 & 1.1830 & 1.2004 \\
         &  & std & 0.6943 & 0.7389 & 0.7028 & 0.8086 & 0.4678 & 0.1223 & 0.2148 & 0.2294 & 0.2514 & 0.1829 \\ \cmidrule(l){2-13}
         & \multirow{2}{*}{0.2} & mean & 6.0830 & 6.5949 & 5.9595 & 5.3552 & 5.0854 & 0.8611 & 1.3713 & 1.4206 & 1.3085 & 1.4601 \\
         &  & std & 0.9913 & 1.0297 & 0.7839 & 1.0495 & 1.1607 & 0.1413 & 0.2572 & 0.1928 & 0.1994 & 0.2775 \\ \midrule
        \multirow{8}{*}{t$(3)$} & \multirow{2}{*}{0} & mean & 2.6800 & 2.9373 & 2.8252 & 2.6418 & 1.8335 & 0.9674 & 1.4802 & 1.2509 & 1.2672 & 1.1694 \\
         &  & std & 0.3519 & 0.3831 & 0.4466 & 0.4975 & 0.6461 & 0.1470 & 0.2671 & 0.1952 & 0.2579 & 0.3124 \\ \cmidrule(l){2-13}
         & \multirow{2}{*}{0.05} & mean & 3.8596 & 4.2828 & 3.8771 & 3.3158 & 2.4803 & 0.9937 & 1.4834 & 1.2301 & 1.0838 & 1.1367 \\
         &  & std & 0.7779 & 0.8233 & 0.6707 & 0.5232 & 0.6834 & 0.1913 & 0.3151 & 0.2207 & 0.1661 & 0.2186 \\ \cmidrule(l){2-13}
         & \multirow{2}{*}{0.1} & mean & 4.8815 & 4.6403 & 5.3096 & 4.5717 & 3.2845 & 1.0613 & 1.0734 & 1.3632 & 1.2919 & 1.1608 \\
         &  & std & 0.9005 & 0.7070 & 0.8156 & 0.9130 & 0.7832 & 0.2049 & 0.2095 & 0.2404 & 0.2373 & 0.2568 \\ \cmidrule(l){2-13}
         & \multirow{2}{*}{0.2} & mean & 7.0219 & 5.9059 & 6.2532 & 5.6690 & 5.2359 & 1.1004 & 1.1926 & 1.1845 & 1.1322 & 1.1964 \\
         &  & std & 0.6926 & 0.6474 & 1.2075 & 0.7614 & 1.5671 & 0.1534 & 0.1457 & 0.1557 & 0.2172 & 0.2420 \\ \midrule
        \multirow{8}{*}{U$(-1, 1)$} & \multirow{2}{*}{0} & mean & 0.8483 & 0.8062 & 0.8519 & 0.8077 & 0.5459 & 1.5189 & 1.3554 & 1.3818 & 1.2370 & 1.0708 \\
         &  & std & 0.1208 & 0.1035 & 0.1036 & 0.1164 & 0.1901 & 0.1802 & 0.2006 & 0.1723 & 0.2021 & 0.3188 \\ \cmidrule(l){2-13}
         & \multirow{2}{*}{0.05} & mean & 1.0764 & 1.0708 & 1.1039 & 0.9470 & 0.5721 & 1.5218 & 1.4023 & 1.3360 & 1.2669 & 1.2674 \\
         &  & std & 0.1627 & 0.1598 & 0.1993 & 0.1467 & 0.2343 & 0.2610 & 0.2763 & 0.2178 & 0.1958 & 0.3629 \\ \cmidrule(l){2-13}
         & \multirow{2}{*}{0.1} & mean & 1.0488 & 1.2561 & 1.1950 & 1.1398 & 0.6686 & 1.4198 & 1.3949 & 1.3189 & 1.2142 & 1.0641 \\
         &  & std & 0.1247 & 0.1890 & 0.1219 & 0.1744 & 0.3164 & 0.1778 & 0.2630 & 0.1933 & 0.2207 & 0.3481 \\ \cmidrule(l){2-13}
         & \multirow{2}{*}{0.2} & mean & 1.5591 & 1.6670 & 1.4182 & 1.2928 & 0.8542 & 1.4387 & 1.3615 & 1.2676 & 1.1630 & 0.9363 \\
         &  & std & 0.2374 & 0.2060 & 0.2295 & 0.2680 & 0.2455 & 0.2512 & 0.2248 & 0.2049 & 0.2124 & 0.1871 \\ \bottomrule
        \end{tabular}}
\end{table}

\begin{table}[t]
    \centering
    \caption{The means and standard derivations of RASE in Example 1b}\label{tab_homodel_nonp_asym1}
    \renewcommand\arraystretch{0.9}
    \resizebox{\textwidth}{!}{%
    \begin{tabular}{@{}ccc|ccccc|ccccc@{}}
        \toprule
        \multirow{2}{*}{Distributions} & \multirow{2}{*}{$\lambda$} &  & \multicolumn{5}{c|}{RASE$(\widehat{m}^*, \widehat{m}_{\text{oll}})$} & \multicolumn{5}{c}{RASE$(\widehat{m}_{\text{alad}},   \widehat{m}_{\text{oll}})$} \\ \cmidrule(l){4-13}
         &  &  & $m = 1$ & $m = 5$ & $m = 10$ & $m = 20$ & $m = 50$ & $m = 1$ & $m = 5$ & $m = 10$ & $m = 20$ & $m = 50$ \\ \midrule
        \multirow{8}{*}{F$(10, 6)$} & \multirow{2}{*}{0} & mean & 2.2105 & 1.9803 & 1.4859 & 0.6743 & 0.2841 & 0.0557 & 0.0671 & 0.0736 & 0.0835 & 0.1435 \\
         &  & std & 0.4652 & 0.3926 & 0.3225 & 0.2268 & 0.0982 & 0.0121 & 0.0112 & 0.0097 & 0.0151 & 0.0272 \\ \cmidrule(l){2-13}
         & \multirow{2}{*}{0.05} & mean & 2.7647 & 2.7262 & 1.9325 & 0.9440 & 0.4545 & 0.0629 & 0.1058 & 0.1099 & 0.1159 & 0.1744 \\
         &  & std & 0.4742 & 0.5051 & 0.5250 & 0.1947 & 0.1416 & 0.0091 & 0.0152 & 0.0209 & 0.0211 & 0.0296 \\ \cmidrule(l){2-13}
         & \multirow{2}{*}{0.1} & mean & 3.1725 & 3.0561 & 1.9151 & 1.1754 & 0.4514 & 0.0833 & 0.1188 & 0.1129 & 0.1497 & 0.2054 \\
         &  & std & 0.8079 & 0.6363 & 0.5131 & 0.2123 & 0.1322 & 0.0115 & 0.0138 & 0.0192 & 0.0243 & 0.0291 \\ \cmidrule(l){2-13}
         & \multirow{2}{*}{0.2} & mean & 4.2862 & 4.5000 & 3.8297 & 2.1891 & 0.4648 & 0.1493 & 0.1694 & 0.1863 & 0.2066 & 0.2965 \\
         &  & std & 0.9300 & 0.6794 & 0.8001 & 0.8241 & 0.1498 & 0.0195 & 0.0189 & 0.0219 & 0.0406 & 0.0419 \\ \midrule
        \multirow{8}{*}{F$(4, 6)$} & \multirow{2}{*}{0} & mean & 2.2923 & 2.0137 & 1.7210 & 0.9068 & 0.5713 & 0.0527 & 0.0829 & 0.0781 & 0.0811 & 0.1262 \\
         &  & std & 0.3799 & 0.3676 & 0.3873 & 0.2522 & 0.2192 & 0.0098 & 0.0153 & 0.0121 & 0.0114 & 0.0215 \\ \cmidrule(l){2-13}
         & \multirow{2}{*}{0.05} & mean & 2.9001 & 2.9828 & 2.5495 & 1.6956 & 0.6319 & 0.0758 & 0.0954 & 0.1050 & 0.1097 & 0.1975 \\
         &  & std & 0.5529 & 0.5718 & 0.5257 & 0.5098 & 0.4732 & 0.0134 & 0.0157 & 0.0241 & 0.0160 & 0.0270 \\ \cmidrule(l){2-13}
         & \multirow{2}{*}{0.1} & mean & 3.4260 & 3.6826 & 3.3611 & 2.3170 & 0.8495 & 0.0973 & 0.1181 & 0.1233 & 0.1587 & 0.1971 \\
         &  & std & 0.6717 & 0.6709 & 0.7558 & 0.6924 & 0.7043 & 0.0142 & 0.0173 & 0.0190 & 0.0208 & 0.0329 \\ \cmidrule(l){2-13}
         & \multirow{2}{*}{0.2} & mean & 4.7415 & 5.3339 & 4.4346 & 2.9963 & 1.2439 & 0.1564 & 0.1973 & 0.2519 & 0.2205 & 0.2836 \\
         &  & std & 0.9762 & 1.0380 & 0.8526 & 0.6186 & 0.9853 & 0.0202 & 0.0181 & 0.0257 & 0.0250 & 0.0518 \\ \midrule
        \multirow{8}{*}{Gamma$(2, 1.5)$} & \multirow{2}{*}{0} & mean & 1.2126 & 1.0855 & 1.0170 & 0.9705 & 0.4944 & 0.0851 & 0.1237 & 0.1310 & 0.1683 & 0.2398 \\
         &  & std & 0.1730 & 0.1423 & 0.1784 & 0.1917 & 0.1356 & 0.0116 & 0.0224 & 0.0190 & 0.0232 & 0.0518 \\ \cmidrule(l){2-13}
         & \multirow{2}{*}{0.05} & mean & 1.6789 & 1.3979 & 1.2808 & 1.1338 & 0.5061 & 0.1081 & 0.1611 & 0.1648 & 0.2019 & 0.2918 \\
         &  & std & 0.2052 & 0.1815 & 0.1720 & 0.1885 & 0.1865 & 0.0232 & 0.0253 & 0.0301 & 0.0292 & 0.0524 \\ \cmidrule(l){2-13}
         & \multirow{2}{*}{0.1} & mean & 1.8913 & 1.8694 & 1.5126 & 1.1989 & 0.5927 & 0.1365 & 0.1909 & 0.2109 & 0.2300 & 0.4111 \\
         &  & std & 0.3129 & 0.2677 & 0.1971 & 0.2493 & 0.1285 & 0.0227 & 0.0337 & 0.0335 & 0.0411 & 0.0717 \\ \cmidrule(l){2-13}
         & \multirow{2}{*}{0.2} & mean & 2.5731 & 2.3831 & 2.1421 & 1.6621 & 0.7365 & 0.1894 & 0.2721 & 0.2886 & 0.3554 & 0.3861 \\
         &  & std & 0.3777 & 0.3625 & 0.3346 & 0.3581 & 0.1402 & 0.0329 & 0.0338 & 0.0482 & 0.0547 & 0.0609 \\ \midrule
        \multirow{8}{*}{Lognorm$(0, 1)$} & \multirow{2}{*}{0} & mean & 2.8140 & 2.2767 & 1.9924 & 1.1168 & 0.3754 & 0.0649 & 0.0687 & 0.0798 & 0.0865 & 0.1114 \\
         &  & std & 0.4291 & 0.3239 & 0.4521 & 0.3239 & 0.0893 & 0.0092 & 0.0069 & 0.0106 & 0.0117 & 0.0179 \\ \cmidrule(l){2-13}
         & \multirow{2}{*}{0.05} & mean & 3.1935 & 2.9592 & 2.9543 & 2.1384 & 0.4910 & 0.0734 & 0.1092 & 0.0979 & 0.1255 & 0.1608 \\
         &  & std & 0.4759 & 0.4886 & 0.5325 & 0.5689 & 0.1425 & 0.0093 & 0.0108 & 0.0152 & 0.0180 & 0.0186 \\ \cmidrule(l){2-13}
         & \multirow{2}{*}{0.1} & mean & 4.1320 & 3.8945 & 3.3716 & 3.7500 & 0.9170 & 0.1149 & 0.1472 & 0.1565 & 0.1683 & 0.2087 \\
         &  & std & 0.5061 & 0.6436 & 0.5206 & 0.9616 & 0.2322 & 0.0155 & 0.0169 & 0.0162 & 0.0170 & 0.0325 \\ \cmidrule(l){2-13}
         & \multirow{2}{*}{0.2} & mean & 4.7578 & 5.1623 & 5.2982 & 4.1367 & 1.7631 & 0.1766 & 0.2548 & 0.2431 & 0.2263 & 0.3051 \\
         &  & std & 0.8538 & 1.0442 & 0.6145 & 0.8390 & 0.7249 & 0.0201 & 0.0202 & 0.0249 & 0.0242 & 0.0317 \\ \bottomrule
        \end{tabular}
        }
\end{table}

\subsubsection{Heteroscedastic model}

We consider the heteroscedastic model from \cite{Kai2010}:
\begin{equation}\label{hemodel_nonp}
    Y= m(X)+ \sigma(X) \varepsilon,
\end{equation}
where the covariate $X$ follows $U(0, 1)$, the regression function is set as $m(x) = x \sin (2 \pi x)$ and the standard deviation function is chosen as $\sigma(x)=\left(2 + \cos(2\pi X)\right)$.
By taking various error distributions, we manage to estimate $m(x)$ for $x\in [0, 1]$.

\begin{itemize}
    \item Example 2a. We consider the model \eqref{hemodel_nonp} with various symmetric error distributions. The means and standard derivations of RASEs are reported in Table~\ref{tab_hemodel_nonp_sym}.
    \item Example 2b. We consider the model \eqref{hemodel_nonp} with various asymmetric error distributions. The means and standard derivations of RASEs are reported in Table~\ref{tab_hemodel_nonp_asym1}.
\end{itemize}

\begin{table}[t]
    \centering
    \caption{The means and standard derivations of RASE in Example 2a}\label{tab_hemodel_nonp_sym}
    \renewcommand\arraystretch{0.9}
    \resizebox{\textwidth}{!}{%
    \begin{tabular}{@{}ccc|ccccc|ccccc@{}}
        \toprule
        \multirow{2}{*}{Distributions} & \multirow{2}{*}{$\lambda$} &  & \multicolumn{5}{c|}{RASE$(\widehat{m}^*,   \widehat{m}_{\text{oll}})$} & \multicolumn{5}{c}{RASE$(\widehat{m}^*,   \widehat{m}_{\text{alad}})$} \\ \cmidrule(l){4-13}
         &  &  & $m = 1$ & $m = 5$ & $m = 10$ & $m = 20$ & $m = 50$ & $m = 1$ & $m = 5$ & $m = 10$ & $m = 20$ & $m = 50$ \\ \midrule
        \multirow{8}{*}{N$(0, 1)$} & \multirow{2}{*}{0} & mean & 0.7267 & 0.6701 & 0.6569 & 0.7217 & 0.6383 & 1.0981 & 1.3545 & 1.0492 & 1.5443 & 1.1269 \\
         &  & std & 0.2870 & 0.2543 & 0.2287 & 0.3842 & 0.4001 & 0.5581 & 0.6338 & 0.3769 & 0.7511 & 0.5214 \\ \cmidrule(l){2-13}
         & \multirow{2}{*}{0.05} & mean & 0.9538 & 0.7062 & 1.0686 & 0.9732 & 0.7053 & 1.1255 & 0.8872 & 1.7691 & 1.3292 & 1.1181 \\
         &  & std & 0.2527 & 0.3438 & 0.5124 & 0.3257 & 0.3936 & 0.4465 & 0.3952 & 0.8875 & 0.6483 & 0.6647 \\ \cmidrule(l){2-13}
         & \multirow{2}{*}{0.1} & mean & 1.0053 & 0.8328 & 1.0265 & 0.9176 & 0.5652 & 1.3582 & 1.2061 & 1.2046 & 1.0695 & 1.1319 \\
         &  & std & 0.3475 & 0.3665 & 0.5707 & 0.3404 & 0.2744 & 0.6471 & 0.6225 & 0.5419 & 0.4794 & 0.6302 \\ \cmidrule(l){2-13}
         & \multirow{2}{*}{0.2} & mean & 1.4315 & 1.1714 & 1.0853 & 1.1144 & 0.6050 & 1.4027 & 1.0070 & 1.2111 & 1.1527 & 1.0223 \\
         &  & std & 0.5522 & 0.5248 & 0.5640 & 0.4442 & 0.3553 & 0.6212 & 0.4598 & 0.5722 & 0.4309 & 0.5522 \\ \midrule
        \multirow{8}{*}{Standard Laplace} & \multirow{2}{*}{0} & mean & 1.4999 & 1.0853 & 0.9633 & 1.0135 & 0.7689 & 1.1310 & 1.0666 & 1.2058 & 1.1943 & 1.0125 \\
         &  & std & 0.5218 & 0.5398 & 0.5117 & 0.5058 & 0.4388 & 0.4377 & 0.5606 & 0.6056 & 0.6562 & 0.5396 \\ \cmidrule(l){2-13}
         & \multirow{2}{*}{0.05} & mean & 1.4697 & 1.5707 & 1.1027 & 1.2700 & 0.7391 & 1.1903 & 1.1582 & 1.0135 & 1.1848 & 1.1693 \\
         &  & std & 0.6377 & 0.8914 & 0.6920 & 0.6536 & 0.4856 & 0.6497 & 0.4998 & 0.5536 & 0.4862 & 0.7794 \\ \cmidrule(l){2-13}
         & \multirow{2}{*}{0.1} & mean & 2.2494 & 1.3594 & 1.3814 & 1.2878 & 0.8373 & 1.3913 & 1.1052 & 1.2642 & 0.9966 & 1.0466 \\
         &  & std & 0.5109 & 0.7009 & 0.7871 & 0.5168 & 0.3933 & 0.7358 & 0.6093 & 0.7675 & 0.4452 & 0.3997 \\ \cmidrule(l){2-13}
         & \multirow{2}{*}{0.2} & mean & 2.6447 & 2.2397 & 1.9286 & 1.3453 & 0.9035 & 1.0531 & 1.5186 & 1.3266 & 1.0250 & 1.2183 \\
         &  & std & 1.2271 & 1.0990 & 1.0036 & 0.6600 & 0.4686 & 0.5970 & 0.7892 & 0.6896 & 0.6771 & 0.7485 \\ \midrule
        \multirow{8}{*}{t$(3)$} & \multirow{2}{*}{0} & mean & 1.4659 & 1.3074 & 1.3026 & 1.1767 & 1.0153 & 0.8626 & 1.3175 & 1.5374 & 1.1319 & 1.0242 \\
         &  & std & 0.5563 & 0.4119 & 0.7490 & 0.5388 & 0.3986 & 0.4161 & 0.5881 & 0.6615 & 0.6002 & 0.4903 \\ \cmidrule(l){2-13}
         & \multirow{2}{*}{0.05} & mean & 1.8777 & 1.5717 & 1.2904 & 1.3036 & 1.1208 & 1.0056 & 1.3165 & 0.9650 & 1.0788 & 1.1889 \\
         &  & std & 0.7718 & 0.8743 & 0.7174 & 0.5199 & 0.4682 & 0.4710 & 0.6300 & 0.5224 & 0.4739 & 0.5791 \\ \cmidrule(l){2-13}
         & \multirow{2}{*}{0.1} & mean & 2.2471 & 2.2166 & 1.7593 & 1.3578 & 0.9391 & 1.0773 & 1.9605 & 1.4028 & 1.1274 & 1.1589 \\
         &  & std & 0.9187 & 0.4275 & 0.7007 & 0.7247 & 0.5182 & 0.4948 & 0.7659 & 0.7115 & 0.4101 & 0.7427 \\ \cmidrule(l){2-13}
         & \multirow{2}{*}{0.2} & mean & 2.3384 & 2.2925 & 2.3013 & 2.0294 & 1.5606 & 1.4250 & 1.4807 & 1.3803 & 1.3728 & 1.3745 \\
         &  & std & 0.8727 & 1.2244 & 1.0566 & 0.9287 & 0.7356 & 0.7656 & 0.7735 & 0.7306 & 0.7876 & 0.6665 \\ \midrule
        \multirow{8}{*}{U$(-1, 1)$} & \multirow{2}{*}{0} & mean & 0.5775 & 0.6357 & 0.5422 & 0.5763 & 0.6284 & 1.2524 & 1.6780 & 1.2995 & 1.2697 & 1.5110 \\
         &  & std & 0.2048 & 0.2156 & 0.1422 & 0.1870 & 0.3707 & 0.4957 & 0.5885 & 0.4593 & 0.5718 & 0.8742 \\ \cmidrule(l){2-13}
         & \multirow{2}{*}{0.05} & mean & 0.6743 & 0.6963 & 0.7637 & 0.6572 & 0.6612 & 1.1182 & 1.2658 & 1.6366 & 1.2909 & 1.4651 \\
         &  & std & 0.2907 & 0.2215 & 0.3464 & 0.2666 & 0.2994 & 0.4560 & 0.5006 & 0.6806 & 0.5737 & 0.8189 \\ \cmidrule(l){2-13}
         & \multirow{2}{*}{0.1} & mean & 0.8104 & 0.8848 & 0.8394 & 0.6523 & 0.6803 & 1.5031 & 1.4863 & 1.3842 & 1.2355 & 1.2383 \\
         &  & std & 0.3892 & 0.2502 & 0.3797 & 0.3503 & 0.4333 & 0.7581 & 0.5698 & 0.7026 & 0.5961 & 0.7495 \\ \cmidrule(l){2-13}
         & \multirow{2}{*}{0.2} & mean & 1.0071 & 0.8618 & 0.8145 & 0.9715 & 0.5718 & 1.6736 & 1.3433 & 1.1674 & 1.2686 & 1.0929 \\
         &  & std & 0.4036 & 0.3721 & 0.3889 & 0.4177 & 0.2240 & 0.4895 & 0.6096 & 0.6170 & 0.6462 & 0.5696 \\ \bottomrule
        \end{tabular}}
\end{table}

\begin{table}[t]
    \centering
    \caption{The means and standard derivations of RASE in Example 2b}\label{tab_hemodel_nonp_asym1}
    \renewcommand\arraystretch{0.9}
    \resizebox{\textwidth}{!}{%
    \begin{tabular}{@{}ccc|ccccc|ccccc@{}}
        \toprule
        \multirow{2}{*}{Distributions} & \multirow{2}{*}{$\lambda$} &  & \multicolumn{5}{c|}{RASE$(\widehat{m}^*,   \widehat{m}_{\text{oll}})$} & \multicolumn{5}{c}{RASE$(\widehat{m}_{\text{alad}}, \widehat{m}_{\text{oll}})$} \\ \cmidrule(l){4-13}
         &  &  & $m = 1$ & $m = 5$ & $m = 10$ & $m = 20$ & $m = 50$ & $m = 1$ & $m = 5$ & $m = 10$ & $m = 20$ & $m = 50$ \\ \midrule
        \multirow{8}{*}{F$(10, 4)$} & \multirow{2}{*}{0} & mean & 1.3735 & 1.4638 & 0.8194 & 1.1962 & 0.1672 & 0.0371 & 0.0405 & 0.0369 & 0.0375 & 0.0667 \\
         &  & std & 0.6105 & 0.4985 & 0.4189 & 0.8522 & 0.1244 & 0.0181 & 0.0163 & 0.0175 & 0.0164 & 0.0351 \\ \cmidrule(l){2-13}
         & \multirow{2}{*}{0.05} & mean & 1.8304 & 1.4430 & 1.2087 & 0.5668 & 0.2145 & 0.0502 & 0.0398 & 0.0445 & 0.0559 & 0.0576 \\
         &  & std & 0.8064 & 0.6330 & 0.7744 & 0.3756 & 0.1457 & 0.0275 & 0.0215 & 0.0234 & 0.0289 & 0.0289 \\ \cmidrule(l){2-13}
         & \multirow{2}{*}{0.1} & mean & 1.7119 & 1.5879 & 1.5812 & 0.8193 & 0.3009 & 0.0305 & 0.0398 & 0.0571 & 0.0503 & 0.0976 \\
         &  & std & 0.5578 & 0.6576 & 0.8008 & 0.4651 & 0.2665 & 0.0140 & 0.0216 & 0.0238 & 0.0258 & 0.0558 \\ \cmidrule(l){2-13}
         & \multirow{2}{*}{0.2} & mean & 2.3697 & 1.6646 & 1.6660 & 1.1322 & 0.1618 & 0.0414 & 0.0627 & 0.0656 & 0.0768 & 0.0811 \\
         &  & std & 0.9219 & 0.7953 & 0.7689 & 0.4644 & 0.1335 & 0.0164 & 0.0237 & 0.0315 & 0.0316 & 0.0366 \\ \midrule
        \multirow{8}{*}{F$(2, 5)$} & \multirow{2}{*}{0} & mean & 1.1049 & 1.0591 & 1.0363 & 0.8626 & 0.3579 & 0.0250 & 0.0294 & 0.0263 & 0.0399 & 0.0472 \\
         &  & std & 0.3335 & 0.5031 & 0.4312 & 0.3238 & 0.1039 & 0.0103 & 0.0103 & 0.0105 & 0.0238 & 0.0204 \\ \cmidrule(l){2-13}
         & \multirow{2}{*}{0.05} & mean & 1.1271 & 1.3010 & 1.0532 & 0.9881 & 0.3566 & 0.0278 & 0.0393 & 0.0377 & 0.0421 & 0.0526 \\
         &  & std & 0.5365 & 0.6589 & 0.4739 & 0.3564 & 0.0675 & 0.0157 & 0.0197 & 0.0179 & 0.0227 & 0.0273 \\ \cmidrule(l){2-13}
         & \multirow{2}{*}{0.1} & mean & 1.2899 & 1.3416 & 1.3446 & 0.9262 & 0.5515 & 0.0309 & 0.0381 & 0.0388 & 0.0434 & 0.0466 \\
         &  & std & 0.3166 & 0.6588 & 0.8187 & 0.4101 & 0.2547 & 0.0151 & 0.0167 & 0.0188 & 0.0165 & 0.0282 \\ \cmidrule(l){2-13}
         & \multirow{2}{*}{0.2} & mean & 1.4844 & 1.4122 & 1.4810 & 1.3220 & 0.3950 & 0.0335 & 0.0530 & 0.0575 & 0.0622 & 0.0665 \\
         &  & std & 0.7717 & 0.6367 & 0.8645 & 0.5650 & 0.1435 & 0.0131 & 0.0243 & 0.0249 & 0.0254 & 0.0324 \\ \midrule
        \multirow{8}{*}{Gamma$(2, 2)$} & \multirow{2}{*}{0} & mean & 1.0588 & 0.9268 & 0.9538 & 0.6224 & 0.2620 & 0.0487 & 0.0613 & 0.0726 & 0.1031 & 0.1561 \\
         &  & std & 0.4150 & 0.3488 & 0.4002 & 0.2660 & 0.1756 & 0.0177 & 0.0286 & 0.0274 & 0.0447 & 0.0669 \\ \cmidrule(l){2-13}
         & \multirow{2}{*}{0.05} & mean & 1.3546 & 1.1785 & 0.9647 & 0.6167 & 0.2957 & 0.0579 & 0.1142 & 0.1327 & 0.1248 & 0.1991 \\
         &  & std & 0.3608 & 0.3919 & 0.3832 & 0.2370 & 0.1486 & 0.0238 & 0.0503 & 0.0677 & 0.0486 & 0.1150 \\ \cmidrule(l){2-13}
         & \multirow{2}{*}{0.1} & mean & 1.3725 & 1.2431 & 1.1622 & 0.6920 & 0.3197 & 0.0910 & 0.0941 & 0.1779 & 0.1460 & 0.1624 \\
         &  & std & 0.4375 & 0.3993 & 0.4439 & 0.3158 & 0.2061 & 0.0407 & 0.0355 & 0.0621 & 0.0582 & 0.0718 \\ \cmidrule(l){2-13}
         & \multirow{2}{*}{0.2} & mean & 1.3237 & 1.2846 & 1.1816 & 0.9477 & 0.3295 & 0.1308 & 0.1518 & 0.1243 & 0.1258 & 0.1850 \\
         &  & std & 0.3888 & 0.4407 & 0.3685 & 0.4449 & 0.1380 & 0.0343 & 0.0673 & 0.0614 & 0.0508 & 0.0769 \\ \midrule
        \multirow{8}{*}{Lognorm$(0.5, 1)$} & \multirow{2}{*}{0} & mean & 1.1939 & 0.9354 & 0.8505 & 0.8633 & 0.2510 & 0.0156 & 0.0235 & 0.0257 & 0.0266 & 0.0359 \\
         &  & std & 0.6830 & 0.4620 & 0.4101 & 0.5293 & 0.1184 & 0.0074 & 0.0084 & 0.0120 & 0.0156 & 0.0128 \\ \cmidrule(l){2-13}
         & \multirow{2}{*}{0.05} & mean & 1.2466 & 1.2355 & 1.1767 & 0.8856 & 0.2059 & 0.0208 & 0.0288 & 0.0309 & 0.0345 & 0.0456 \\
         &  & std & 0.9349 & 0.7200 & 0.6983 & 0.3891 & 0.1105 & 0.0057 & 0.0117 & 0.0137 & 0.0133 & 0.0152 \\ \cmidrule(l){2-13}
         & \multirow{2}{*}{0.1} & mean & 1.2581 & 1.4572 & 1.4016 & 1.1773 & 0.2357 & 0.0233 & 0.0314 & 0.0396 & 0.0407 & 0.0599 \\
         &  & std & 0.7062 & 0.6170 & 0.6901 & 0.6054 & 0.1108 & 0.0054 & 0.0140 & 0.0169 & 0.0120 & 0.0143 \\ \cmidrule(l){2-13}
         & \multirow{2}{*}{0.2} & mean & 1.9738 & 1.9726 & 1.8643 & 1.2611 & 0.4986 & 0.0349 & 0.0415 & 0.0458 & 0.0449 & 0.0540 \\
         &  & std & 0.8356 & 0.8191 & 0.9907 & 0.6841 & 0.2230 & 0.0136 & 0.0176 & 0.0149 & 0.0166 & 0.0237 \\ \bottomrule
        \end{tabular}}
\end{table}

From the numerical results in Tables~\ref{tab_hemodel_nonp_sym} - \ref{tab_hemodel_nonp_asym1}, we have the following observations:

(i) In general, the results in this example are similar to those in example 1. Compared with the results in example 1, however, the superiority of $\widehat{m}^*(\cdot)$ over $\widehat{m}_{\text{oll}}\br{\cdot}$ is slightly decreased.

(ii) Particularly, our estimator $\widehat{m}^*(\cdot)$ generally performs better than $\widehat{m}_{\text{oll}}\br{\cdot}$ in cases with smaller $m$ and larger $\lambda$, illustrating the robustness of our estimator. On the other hand, our estimator $\widehat{m}^*(\cdot)$ generally outperforms $\widehat{m}_{\text{alad}}\br{\cdot}$, and in the case with  asymmetric error distributions, the estimator $\widehat{m}_{\text{alad}}\br{\cdot}$ is even inconsistent.
When the number $m$ is larger than $20$, however, our estimator $\widehat{m}^*(\cdot)$ is not better than the benchmark estimator $\widehat{m}_{\text{oll}}$ because the values of RASE$(\widehat{m}^*, \widehat{m}_{\text{oll}})$ is generally less than $1$. This is not surprising since
the benchmark estimator $\widehat{m}_{\text{oll}}$ is computed on the full data set.

\subsection{Real data examples}

For case study, we apply our method to the Beijing Multi-Site Air-Quality Data set from the UCI machine learning repository\footnote{https://archive-beta.ics.uci.edu/ml/datasets/beijing+multi+site+air+quality+data}.
This data set consists of hourly data about 6 main air pollutants and 6 relevant meteorological variables collected from 12 nationally-controlled air-quality monitoring sites in Beijing, China.
The observational data cover the time period from March 1st, 2013 to February 28th, 2017, and each variable including 420768 observed values.
Our goal is to fit the relationship between the main air pollutants and the relevant meteorological variables in the dataset.

In model~\eqref{eq_model}, we use $Y$ and $X$ to denote a main index from 6 air pollutants and a covariate from 6 relevant meteorological variables  in the data set, respectively.
To test the performance of the estimators, we drop the data that suffer from data missing and then equally divide the remainder data set into training set $\D_{\text{train}}$ and testing set $\D_{\text{test}}$.
In the training set $\D_{\text{train}}$, we call the observations $\br{X_i, Y_i}$ as outliers if $\abs{Y_i - \widehat{m}_{\text{nw}}\br{X_i}} > \gamma \widehat{\sigma}\br{X_i}$,
where $\widehat{m}_{\text{nw}}\br{\cdot}$ and $\widehat{\sigma}\br{\cdot}$ are the pilot estimators from \eqref{eq_pilmnw} and \eqref{eq_pilsgmnw}, respectively, and $\gamma$ is successively taken as $2.5, 3.0$ and $3.5$.
Then the proportion of outliers is defined as $r_{\text{ol}} = n_{\text{ol}}/ n_{\text{train}} \times 100 \%$, where $n_{\text{ol}}$ and $n_{\text{train}}$ are the numbers of outliers and the number of total observations in $\D_{\text{train}}$, respectively.
To show the robustness of estimators, we replace the outliers $(X_i, Y_i)$ by $(X_i, cY_i)$ for the constant $c = \text{NaN}, 1, 2, 5, 10, 50$, successively,
where $c = \text{NaN}$ means that the outliers are removed from the training set.
To simulate the DC environment, the training set is naturally divided into $12$ parts, i.e., $\D_1, \cdots, \D_{12}$, corresponding to the $12$ monitoring sites mentioned before.

In this example, we still consider the estimators $\widehat{m}^*\br{\cdot}$, $\widehat{m}_{\text{alad}}\br{\cdot}$ and $\widehat{m}_{\text{oll}}\br{\cdot}$ computed on the training set $\D_{\text{train}}$.
The prediction accuracy of an estimator $\widehat{g}(\cdot)$ is described by the root mean square error (RMSE) and the mean absolute error (MAE) on $\D_{\text{test}}$, namely
\begin{equation*}
    \text{RMSE}(\widehat{g}) = \sqrt{\frac{1}{n_{\text{test}}} \sum_{(X_i, Y_i) \in \D_{\text{test}}}\br{Y_i - \widehat{g}\br{X_i}}^2},  \quad \text{MAE}(\widehat{g}) = \frac{1}{n_{\text{test}}} \sum_{(X_i, Y_i) \in \D_{\text{test}}}\abs{Y_i - \widehat{g}\br{X_i}},
\end{equation*}
where $n_{\text{test}}$ is the number of observations in $\D_{\text{test}}$.

\subsubsection{Asymmetric Data}

We first fit the relationship between the PM2.5 concentration and the wind speed.
After dropping missing data, the training set and testing set both consist of 205856 observations, which are shown in Figure~\ref{fig_data_pm}.
The Figures~\ref{fig_traina_pm} - \ref{fig_trainc_pm} present the same training set with different values of $\gamma$ (i.e., with different representations of outliers), and Figure~\ref{fig_testdata_pm} shows the scatters of testing data. It can be seen from the figures that the distribution of testing data is very similar to that of the training data.
By the data scatters in Figure~\ref{fig_data_pm}, we conclude  that the random error is strongly asymmetric in this example.

\begin{figure}
    \centering
    \subfigure[Training set with $\gamma = 3.5$]{
        \label{fig_traina_pm}
        \includegraphics[width=0.4\textwidth]{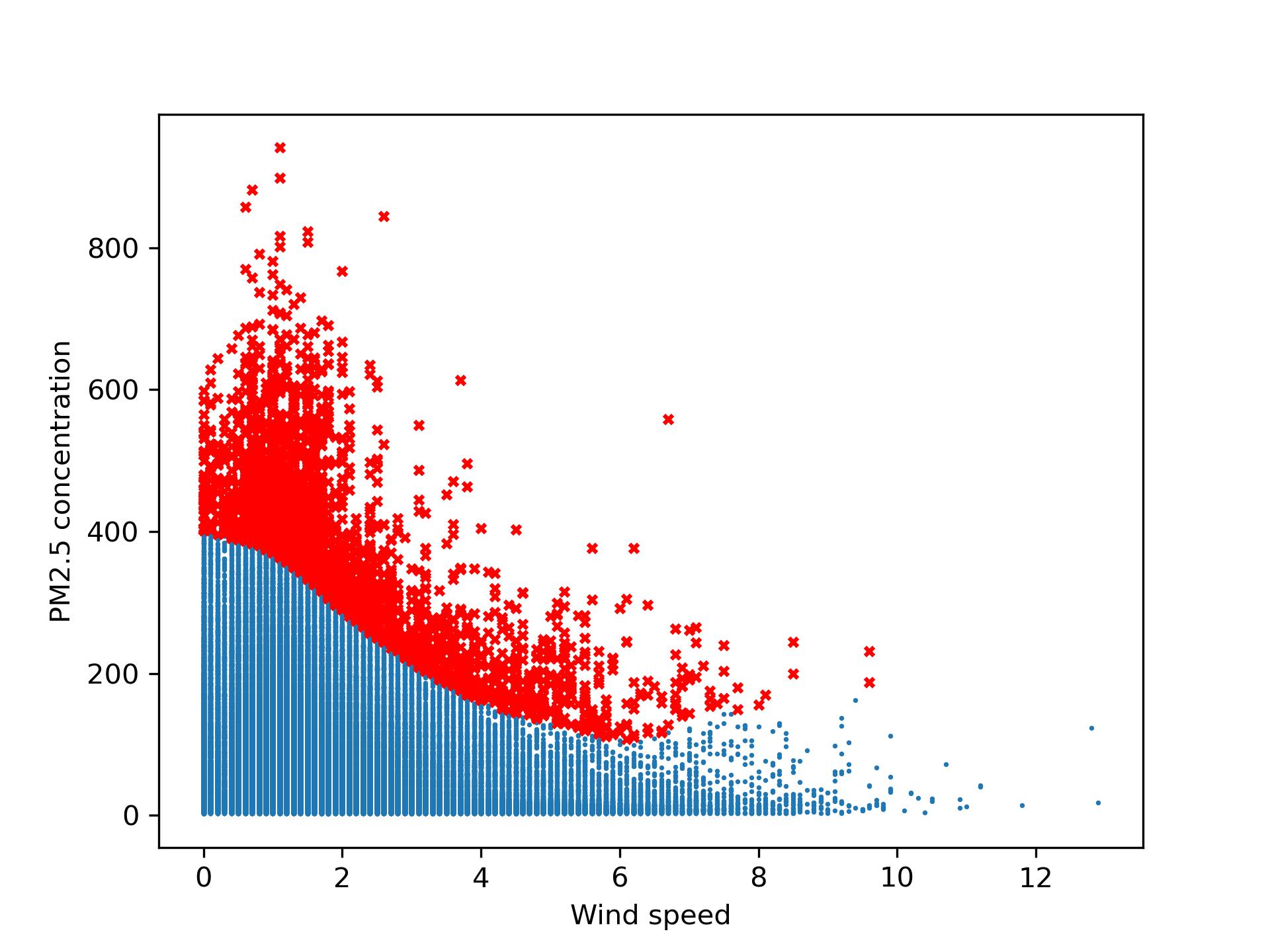}}
    \subfigure[Training set with $\gamma = 3.0$]{
        \label{fig_trainb_pm}
        \includegraphics[width=0.4\textwidth]{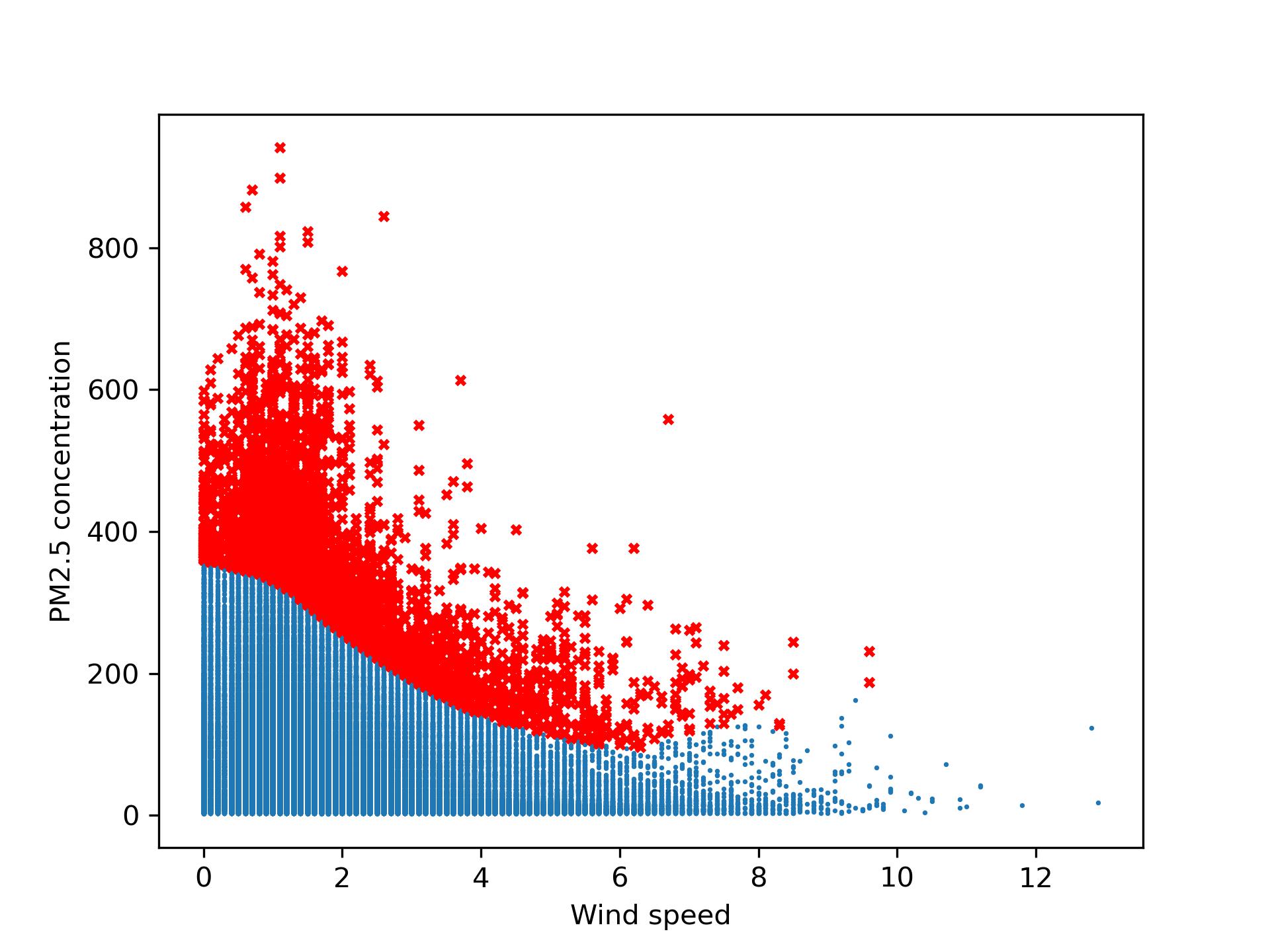}}
    \subfigure[Training set with $\gamma = 2.5$]{
        \label{fig_trainc_pm}
        \includegraphics[width=0.4\textwidth]{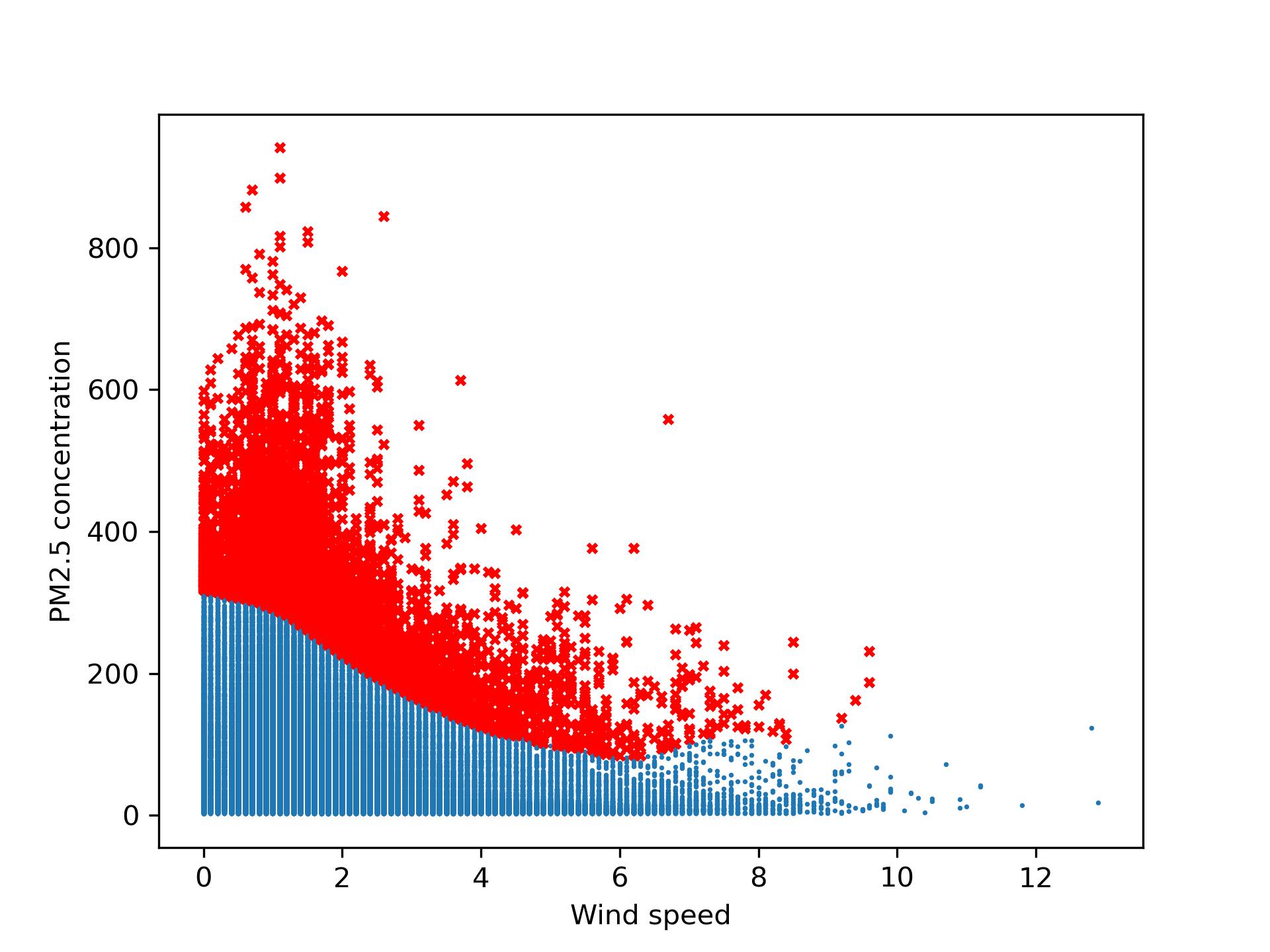}}
    \subfigure[Test set]{
    \label{fig_testdata_pm}
    \includegraphics[width=0.4\textwidth]{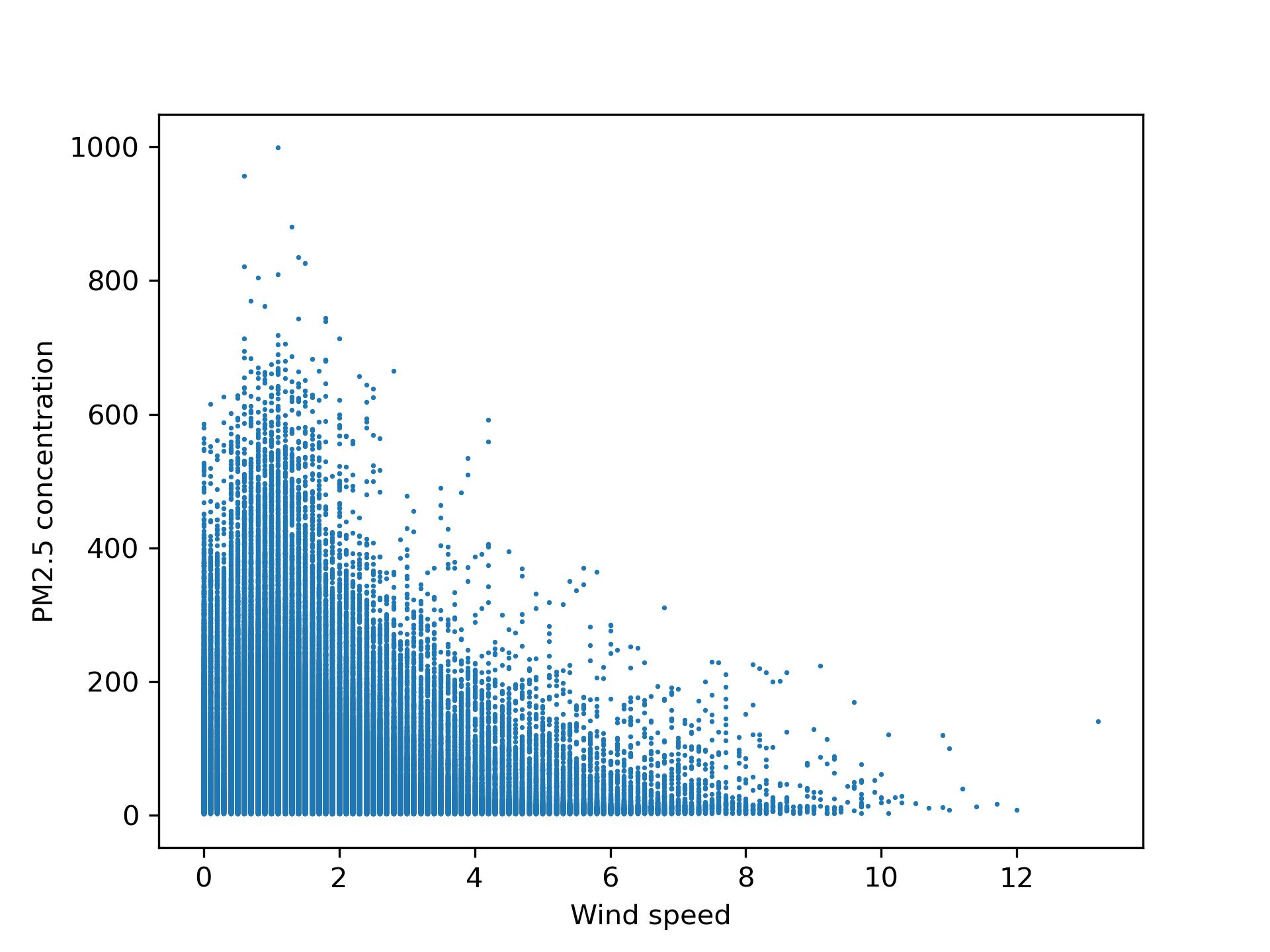}}
    \caption{The scatter plot of data sets the PM2.5 concentration ($\text{ug}/\text{m}^3$) and the wind speed (m/s), where the red signs ``{\color{red} $\pmb{\times}$}'' are outliers in the training set.}\label{fig_data_pm}
\end{figure}

The RMSEs and MAEs of the three estimators are reported in Table~\ref{tab_realdata_pm}.
From Table~\ref{tab_realdata_pm}, we can see that
when $c =1$ and $r_{\text{ol}} = 0$, i.e., without treating the outliers, the estimators $\widehat{m}^*\br{\cdot}$ and $\widehat{m}_{\text{oll}}\br{\cdot}$ appear to have the similar estimation errors, and their RMSEs is lower than that of $\widehat{m}_{\text{alad}}\br{\cdot}$.
As $c$ and $r_{\text{ol}}$ increase, i.e., the influence of outliers becomes significant, the estimation errors of $\widehat{m}^*\br{\cdot}$ and $\widehat{m}_{\text{alad}}\br{\cdot}$ remains nearly the same,
while the estimation errors of local linear estimator $\widehat{m}_{\text{oll}}\br{\cdot}$ increase significantly.
Hence, our estimator possesses strong robustness comparable to the ALAD estimator and is much batter than the local linear estimator.

Regardless of the change in outliers, our estimator gives better RMSEs than that given by $\widehat{m}_{\text{alad}}\br{\cdot}$.
Since the random error in this example is asymmetrical, the estimation of LAD is biased, resulting a poor performance of the ALAD estimator.
However, the asymmetry does not influence our estimator, since estimation bias is corrected by the quantile-matched composite.
We should remark that $\widehat{m}_{\text{alad}}\br{\cdot}$ always has the lowest MAEs among these estimators.
This is not surprising because the LAD estimator is just the minimizer of the MAE on the training set.
This phenomenon is consistent with the results in real data example of \cite{Rong2018Composite}.


\begin{table}
    \centering
    \caption{The performance of estimators in fitting the relationship between the PM2.5 concentration and the wind speed.}
    \label{tab_realdata_pm}
    \renewcommand\arraystretch{0.9}
    \resizebox{0.9\textwidth}{!}{%
    \begin{tabular}{@{}ccc|ccc|ccc@{}}
        \toprule
        \multirow{2}{*}{$c$} & \multirow{2}{*}{$\gamma$} & \multirow{2}{*}{$r_{\text{ol}}$} & \multicolumn{3}{c|}{RMSE} & \multicolumn{3}{c}{MAE} \\ \cmidrule(l){4-9}
         &  &  & $\widehat{m}^*$ & $\widehat{m}_{\text{lad}}$ & $\widehat{m}_{\text{oll}}$ & $\widehat{m}^*$ & $\widehat{m}_{\text{lad}}$ & $\widehat{m}_{\text{oll}}$ \\ \midrule
         \multirow{3}{*}{NaN} & 3.5 & 0.41\% & 77.3222 & 79.6549 & 77.0962 & 53.6384 & 52.3926 & 54.1895 \\
         & 3.0 & 0.64\% & 77.5062 & 79.8907 & 77.2698 & 53.3794 & 52.3790 & 53.8054 \\
         & 2.5 & 1.00\% & 77.8519 & 80.2599 & 77.6102 & 53.0411 & 52.3745 & 53.3455 \\ \midrule
         1 & $\infty$ & 0 & 77.0365 & 79.2596 & 76.9024 & 54.3319 & 52.4481 & 55.3774 \\ \midrule
        \multirow{3}{*}{2} & 3.5 & 0.41\% & 77.0159 & 79.2587 & 77.1191 & 54.4273 & 52.4486 & 57.1707 \\
         & 3.0 & 0.64\% & 77.0193 & 79.2580 & 77.3650 & 54.4135 & 52.4486 & 58.0138 \\
         & 2.5 & 1.00\% & 77.0188 & 79.2588 & 77.8787 & 54.4282 & 52.4485 & 59.3164 \\ \midrule
        \multirow{3}{*}{5} & 3.5 & 0.41\% & 77.0170 & 79.2605 & 80.7022 & 54.5204 & 52.4482 & 64.5048 \\
         & 3.0 & 0.64\% & 77.0255 & 79.2610 & 84.5667 & 54.4697 & 52.4481 & 69.7040 \\
         & 2.5 & 1.00\% & 77.0258 & 79.2644 & 92.0792 & 54.5605 & 52.4476 & 78.3188 \\ \midrule
        \multirow{3}{*}{10} & 3.5 & 0.41\% & 77.0750 & 79.2632 & 94.9023 & 54.5679 & 52.4478 & 81.9343 \\
         & 3.0 & 0.64\% & 77.0761 & 79.2677 & 110.7647 & 54.5348 & 52.4472 & 98.2934 \\
         & 2.5 & 1.00\% & 77.0992 & 79.2695 & 137.9665 & 54.8022 & 52.4470 & 124.9198 \\ \midrule
        \multirow{3}{*}{50} & 3.5 & 0.41\% & 77.5585 & 79.2695 & 314.2545 & 54.1327 & 52.4470 & 303.0136 \\
         & 3.0 & 0.64\% & 77.5669 & 79.2719 & 442.7426 & 54.1374 & 52.4465 & 431.5009 \\
         & 2.5 & 1.00\% & 77.9308 & 79.2751 & 630.4774 & 54.6406 & 52.4460 & 615.5673 \\ \bottomrule
        \end{tabular}
        }
\end{table}

\subsubsection{Symmetric Data}

We manage to fit the relationship between the O3 concentration ($\text{ug}/\text{m}^3$) and the temperature (degree Celsius).
Dropping missing data, the resulting training set and test set both consist of 203550 observations, which are shown in Figure~\ref{fig_data_o3}.
Comparing Figures~\ref{fig_data_pm} and \ref{fig_data_o3}, we can see that in this example, the asymmetry in error distribution is less significant than that in the former example of asymmetric data.
Also, the tail of error distribution is thinner with less outliers.
This example is designed to show the adaptability of our estimator for symmetric errors, i.e. whether or not our estimator has advantages over the competing estimators when the error is symmetric.

\begin{figure}
    \centering
    \subfigure[Training set with $\gamma = 3.5$]{
        \label{fig_traina_o3}
        \includegraphics[width=0.4\textwidth]{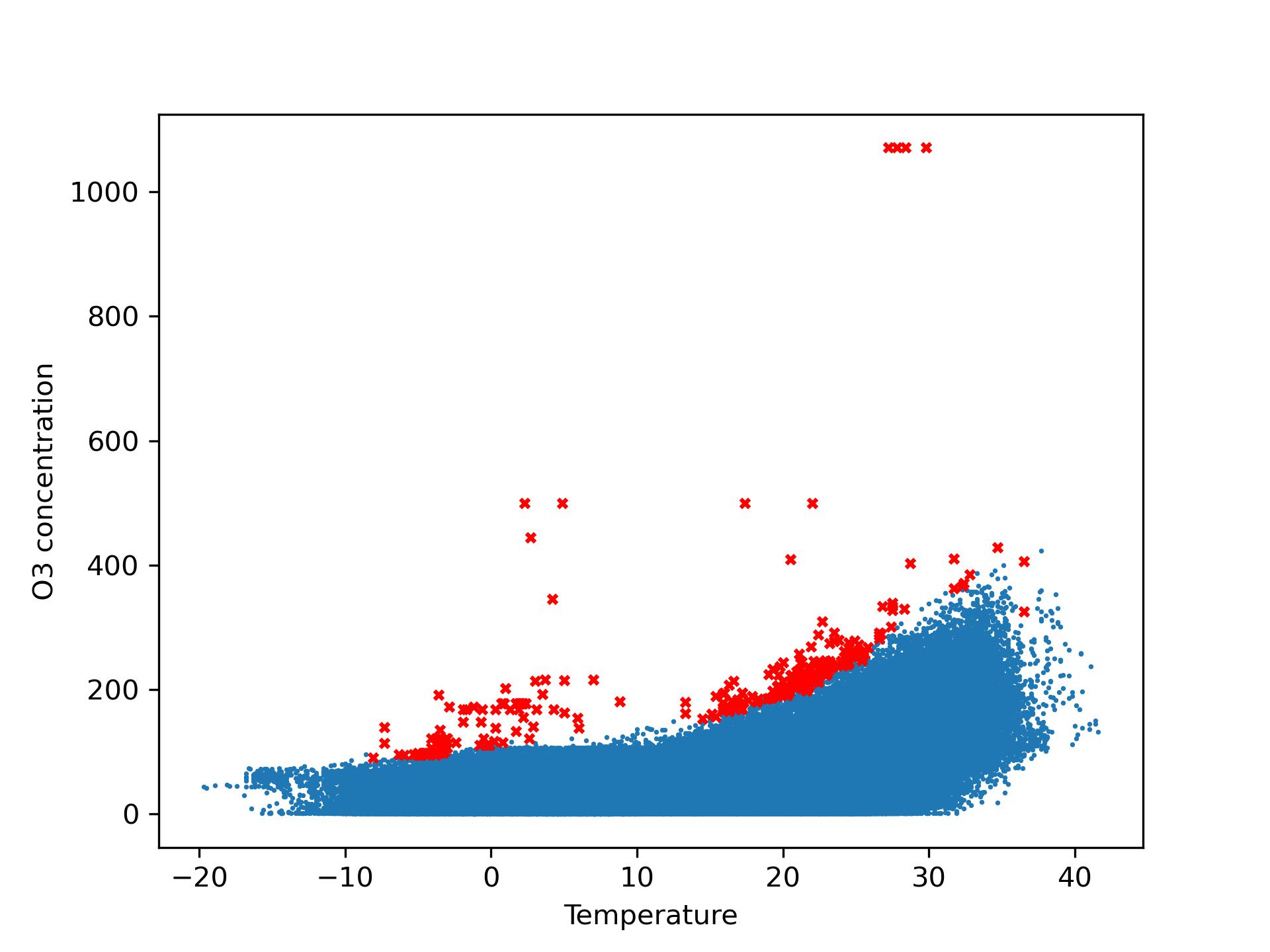}}
    \subfigure[Training set with $\gamma = 3.0$]{
        \label{fig_trainb_o3}
        \includegraphics[width=0.4\textwidth]{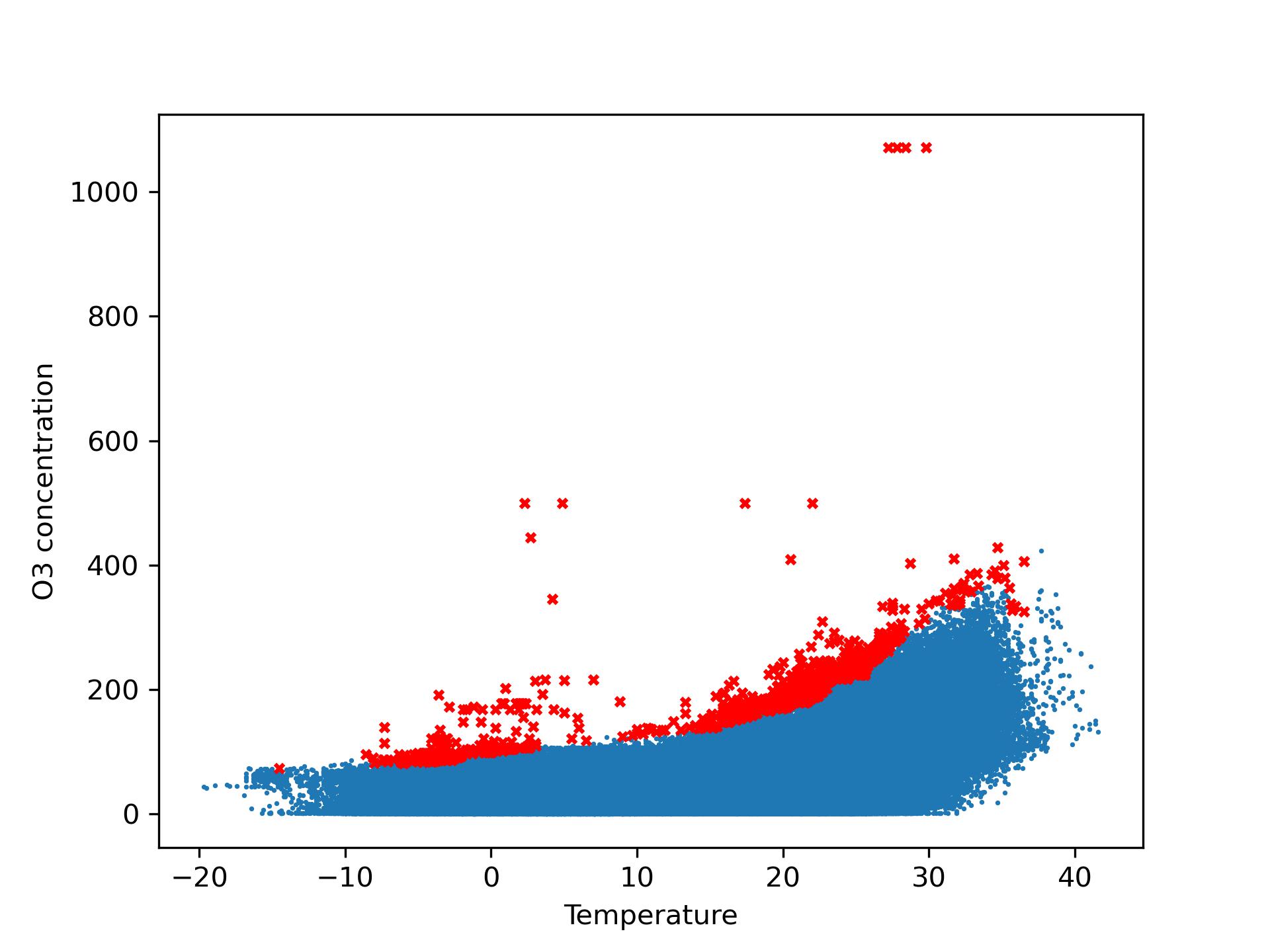}}
    \subfigure[Training set with $\gamma = 2.5$]{
        \label{fig_trainc_o3}
        \includegraphics[width=0.4\textwidth]{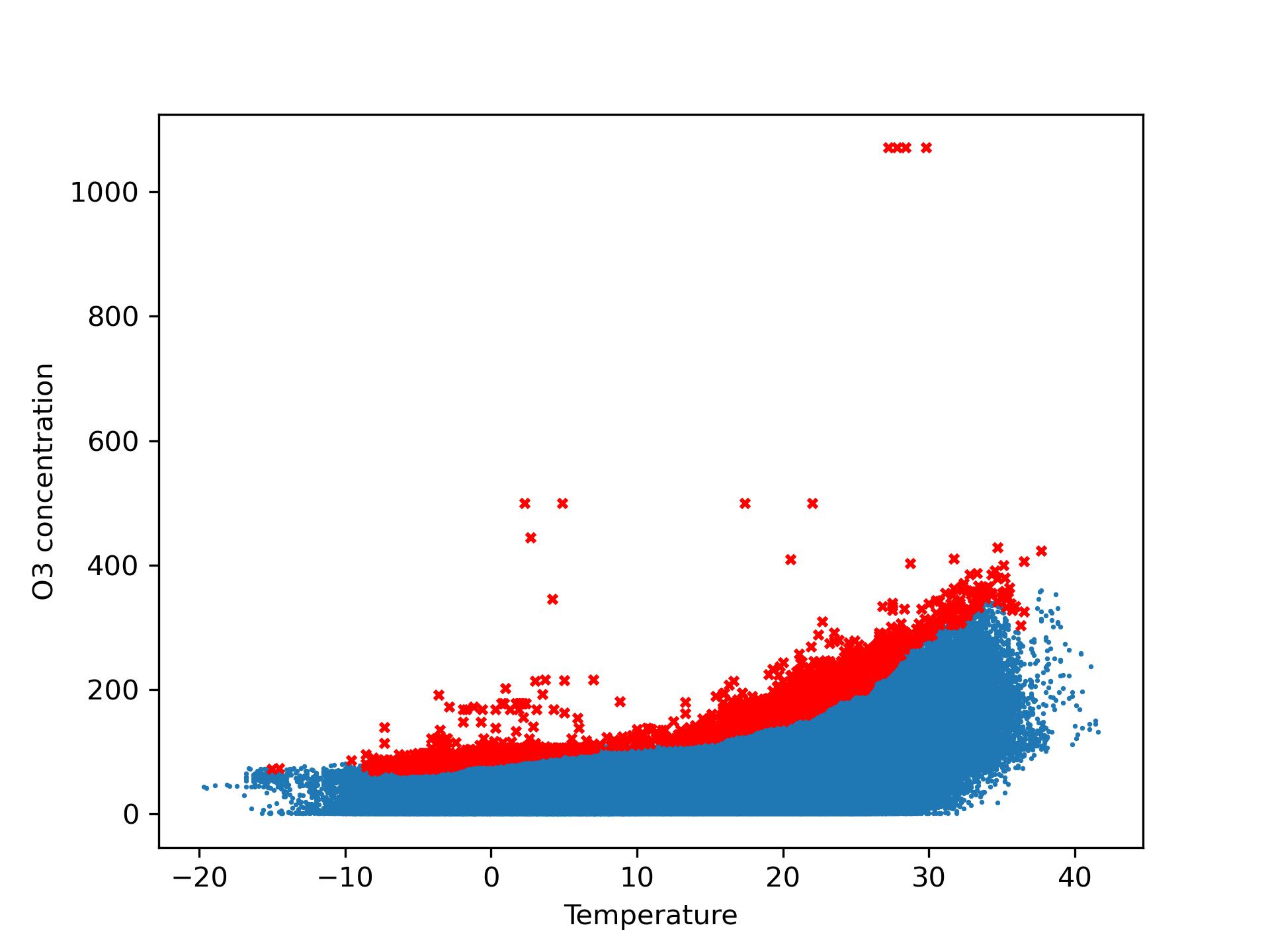}}
    \subfigure[Test set]{
    \label{fig_testdata_o3}
    \includegraphics[width=0.4\textwidth]{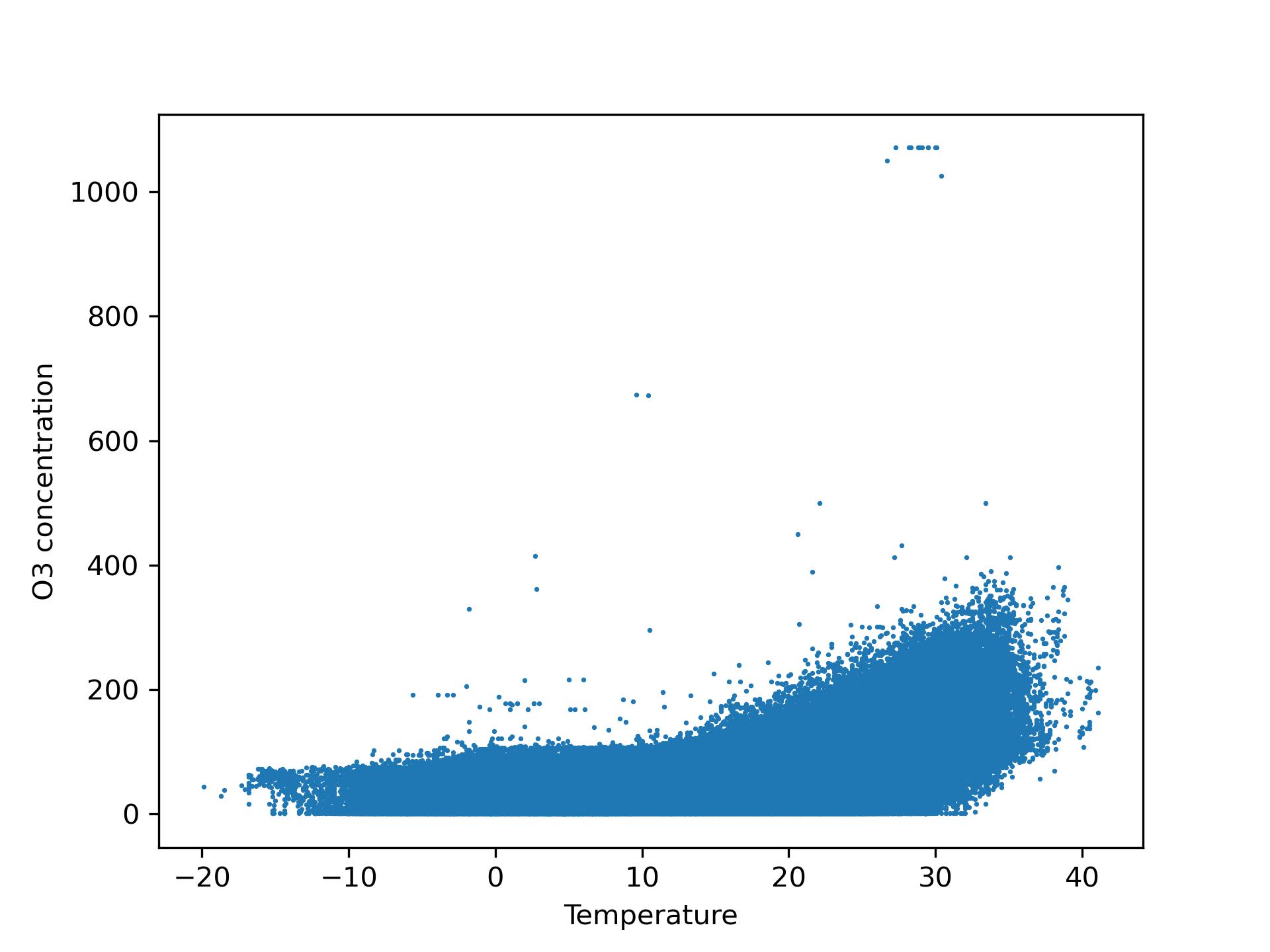}}
    \caption{The scatter plot of the O3 concentration ($\text{ug}/\text{m}^3$) and the temperature (degree Celsius), where the red signs ``{\color{red} $\pmb{\times}$}'' are outliers in the training set.}\label{fig_data_o3}
\end{figure}

The RMSEs and MAEs of the three estimators are listed in Table~\ref{tab_realdata_o3}.
In general, the results in this example are similar to those in the previous example, namely, our estimator together with the ALAD estimator are quite robust against the outliers, which outperform the local linear estimator.
The differences lie in the gap between the RMSEs of $\widehat{m}^*\br{\cdot}$ and $\widehat{m}_{\text{alad}}$, which is smaller than those in the previous example because the random error in this example is more symmetric than before.
Nevertheless, our estimator still shows advantages over the ALAD estimator under the criterion of RMSEs.

In summary, by comprehensively investigating the numerical results in various experiment conditions such as symmetric and asymmetric errors, homoscedastic and heteroscedastic models, homogeneous and heterogeneous data,
we can conclude the effectiveness and the robustness of our estimation method.


\begin{table}
    \centering
    \caption{The performance of estimators in fitting the relationship between the O3 concentration and the temperature.}
    \label{tab_realdata_o3}
    \renewcommand\arraystretch{0.9}
    \resizebox{0.9\textwidth}{!}{%
    \begin{tabular}{@{}ccc|ccc|ccc@{}}
        \toprule
        \multirow{2}{*}{$c$} & \multirow{2}{*}{$\gamma$} & \multirow{2}{*}{$r_{\text{ol}}$} & \multicolumn{3}{c|}{RMSE} & \multicolumn{3}{c}{MAE} \\ \cmidrule(l){4-9}
         &  &  & $\widehat{m}^*$ & $\widehat{m}_{\text{lad}}$ & $\widehat{m}_{\text{oll}}$ & $\widehat{m}^*$ & $\widehat{m}_{\text{lad}}$ & $\widehat{m}_{\text{oll}}$ \\ \midrule
        \multirow{3}{*}{NaN} & 3.5 & 0.03\% & 41.0046 & 41.3196 & 40.9419 & 31.8436 & 31.3884 & 31.8251 \\
         & 3.0 & 0.11\% & 40.9941 & 41.3396 & 40.9469 & 31.8124 & 31.3866 & 31.7951 \\
         & 2.5 & 0.34\% & 40.9782 & 41.4136 & 40.9791 & 31.7313 & 31.3818 & 31.7222 \\ \midrule
        1 & $\infty$ & 0 & 41.0091 & 41.3120 & 40.9407 & 31.8569 & 31.3898 & 31.8466 \\ \midrule
        \multirow{3}{*}{2} & 3.5 & 0.03\% & 41.0105 & 41.3119 & 40.9423 & 31.8578 & 31.3896 & 31.8759 \\
         & 3.0 & 0.11\% & 41.0119 & 41.3121 & 40.9555 & 31.8593 & 31.3897 & 31.9304 \\
         & 2.5 & 0.34\% & 41.0319 & 41.3118 & 41.0314 & 31.8995 & 31.3896 & 32.1106 \\ \midrule
        \multirow{3}{*}{5} & 3.5 & 0.03\% & 41.0141 & 41.3120 & 40.9684 & 31.8624 & 31.3897 & 31.9771 \\
         & 3.0 & 0.11\% & 41.0079 & 41.3120 & 41.1679 & 31.8485 & 31.3897 & 32.2897 \\
         & 2.5 & 0.34\% & 41.1210 & 41.3118 & 42.3723 & 31.8985 & 31.3896 & 33.5672 \\ \midrule
        \multirow{3}{*}{10} & 3.5 & 0.03\% & 41.0091 & 41.3122 & 41.0831 & 31.8486 & 31.3896 & 32.1918 \\
         & 3.0 & 0.11\% & 41.0057 & 41.3122 & 42.0706 & 31.8414 & 31.3896 & 33.2458 \\
         & 2.5 & 0.34\% & 41.2199 & 41.3113 & 47.7517 & 31.8348 & 31.3895 & 38.0221 \\ \midrule
        \multirow{3}{*}{50} & 3.5 & 0.03\% & 41.0097 & 41.3122 & 45.0173 & 31.8508 & 31.3896 & 35.8894 \\
         & 3.0 & 0.11\% & 41.0068 & 41.3117 & 66.6623 & 31.8445 & 31.3896 & 52.5894 \\
         & 2.5 & 0.34\% & 41.4667 & 41.3116 & 139.9563 & 32.3168 & 31.3894 & 114.4186 \\ \bottomrule
        \end{tabular}
        }
\end{table}

\clearpage
\bibliographystyle{abbrvnat}
\bibliography{ref_file}

\newpage
\appendix

\section{Lemmas and Proofs}


\begin{lemma}\label{lemm_nonp}
    Under Assumption~\ref{assum_fFK} and \ref{assum_K}, we have
    \begin{equation}\label{eq_Rsht}
        R^*(h, \tau) = \beta(x, \tau) h^2 + o(h^2)
    \end{equation}
    with $\beta(x, \tau) = \frac{1}{2} \left(m^{\prime\prime}(x) + \sigma^{\prime\prime}(x) F_{\varepsilon}^{-1}(\tau)\right)\int_{\R} v^2 K(v) dv$.
    Consequently, $R^*(h, \tau) = O(h^2)$ uniformly for $\tau \in [\delta_{\tau}, 1-\delta_{\tau}]$ with $\delta_{\tau}$ an arbitrary small constant in $(0, \frac{1}{2})$.
\end{lemma}

\begin{proof}

    The equation~\eqref{eq_Rsht} is a direct corollary of Theorem 1 in \cite{Fan1994Robust}.
    By Theorem 1 in \cite{Guerre2012}, we conclude that $R^*(h, \tau) = O(h^2)$ uniformly for $\tau \in [\delta_{\tau}, 1-\delta_{\tau}]$ with $\delta_{\tau}$ an arbitrary small constant in $(0, \frac{1}{2})$.

\end{proof}

\begin{lemma}\label{lemm_cr_sgm}
    Under Assumption~\ref{assum_fFK} and \ref{assum_bat_bw}, the asymptotic variance in \eqref{eq_sigma} satisfies
    \begin{equation*}
        \Sigma(x, \B{\omega}, \B{\tau}, \B{h}) = \Omega\br{\|\B{\omega}\|_2^2 n^{-s\br{1-\nu}}},
    \end{equation*}
    uniformly for $x$, $\B{\tau}$, $n$ and $m$.
\end{lemma}

\begin{proof}

    By Assumption~\ref{assum_fFK} and \ref{assum_K}, there exist positive constant $\delta_a$ and $M_a$ independent with $m$ and $n$, such that $ \delta_a < a(x) < M_a$.
    Then recalling \eqref{eq_sigma}, using the equation \eqref{eq_lamR}, we have
    \begin{equation*}
        \delta_{\Sigma} \sum_{i=1}^m \sum_{j=1}^J\frac{\omega_{ij}^2}{n_i h_{ij}} \leq \Sigma(x, \B{\omega}, \B{\tau}, \B{h}) \leq M_{\Sigma} \sum_{i=1}^m \sum_{j=1}^J\frac{\omega_{ij}^2}{n_i h_{ij}}
    \end{equation*}
    with $\delta_{\Sigma} = \delta_a\lambda_{\min}$ and $M_{\Sigma} = M_a\lambda_{\max}$ independent with $m$ and $n$.
    By Assumption~\ref{assum_bat_bw}, we have
    \begin{equation*}
        \Sigma(x, \B{\omega}, \B{\tau}, \B{h}) = \Omega\br{\|\B{\omega}\|_2^2 n^{-s\br{1-\nu}}}.
    \end{equation*}
    We complete the proof.

\end{proof}

\subsection{Proof of Theorem~\ref{them_asyproperty}}

\begin{proof}
    Following Theorem~2 in \cite{Guerre2012}, under Assumption~\ref{assum_fFK} and \ref{assum_K}, we have the Bahadur representation
    \begin{equation}\label{eqpf_Bahadur}
        \widehat{m}_i(x; \tau) = Q_{Y, h}^*(\tau|x) + \phi(x, \tau, h, \D_i) + R^B_{i}(n_i, h, \tau),
    \end{equation}
    where
    \begin{equation}\label{eq_RiB}
        R^B_i(n_i, h) = O_p\left(\left(\frac{\log n_i}{n_i h}\right)^{3 / 4}\right),
    \end{equation}
    which holds uniformly for $x \in [0, 1]$ and $\tau \in [\delta_{\tau}, 1-\delta_{\tau}]$.
    By Lemma~\ref{lemm_nonp}, we have
    \begin{equation}\label{eq_QandQstar}
        Q_{Y, h}^*(\tau|x) = m(x) + \sigma(x) F_{\varepsilon}^{-1} (\tau) + \beta(x, \tau) h^2 + o(h^2),
    \end{equation}
    uniformly for $x \in [0, 1]$ and $\tau \in [\delta_{\tau}, 1 - \delta_{\tau}]$,  where
    \begin{equation}
        \beta(x, \tau) = \frac{\mu_2}{2} \left(m^{\prime\prime}(x) + \sigma^{\prime\prime}(x) F_{\varepsilon}^{-1}(\tau)\right).
    \end{equation}
    Inserting \eqref{eq_QandQstar} into \eqref{eqpf_Bahadur}, and using Assumption~\ref{assum_wtau}, we obtain
    \begin{equation}\label{eqpf_hatmx}
        \begin{aligned}
            \widehat{m}(x) =\;& m(x) + B(x, \B{\omega}, \B{\tau}, \B{h}) + \sum_{i=1}^m\sum_{j=1}^J \omega_{ij}\phi(x, \tau_{ij}, h_{ij}, \D_i) \\
            &+ \sum_{i=1}^m\sum_{j=1}^J \omega_{ij} R^B_i(n_i, h_{ij}) + o\br{\sum_{i=1}^m\sum_{j=1}^J \abs{\omega_{ij}} h_{ij}^2}.
        \end{aligned}
    \end{equation}
    By Assumption~\ref{assum_bat_bw}, it holds that
    \begin{equation}\label{eq_sum_wni}
        \left(\frac{\log n_i}{n_i h_{ij}}\right)^{3 / 4} = O\br{\br{\log n}^{\frac{3}{4}} n^{-\frac{3}{4} s \br{1-\nu}}}
    \end{equation}
    \begin{equation}\label{eq_sum_wh}
        \sum_{i=1}^m\sum_{j=1}^J \abs{\omega_{ij}} h_{ij}^2 = O\br{\|\B{\omega}\|_1 n^{-2s\nu}},
    \end{equation}
    uniformly with respect to $i$ and $m$.
    Since all the data batches $\left\{\D_i\right\}_{i=1}^m$ are i.i.d., which means that \eqref{eq_RiB} holds uniformly for $x \in [0, 1]$, $\tau \in [\delta_{\tau}, 1-\delta_{\tau}]$, $i = 1, \cdots, m$ and $m \in \N$.
    Then by \eqref{eq_RiB} and \eqref{eq_sum_wni}, we have
    \begin{equation}\label{eq_sum_rbi}
        \abs{\sum_{i=1}^m\sum_{j=1}^J \omega_{ij} R^B_i(n_i, h)} \leq \left\|\B{\omega}\right\|_2 \sqrt{ J\sum_{i=1}^m \abs{R^B_i(n_i, h)}^2} = O_p\br{\|\B{\omega}\|_2 \br{\log n}^{\frac{3}{4}} n^{-\frac{3}{4} s \br{1-\nu}} n^{\frac{1}{2}(1-s)}}.
    \end{equation}
    Then the equation \eqref{eq_them_asyproperty} is proven by combining \eqref{eq_sum_wh} and \eqref{eq_sum_rbi} with \eqref{eqpf_hatmx}.

    In the following, we prove \eqref{lim_sqrtn_nonp}.
    Letting
    \begin{equation*}
        \xi_{ijk} = \frac{1}{h_{ij} v\left(x; \tau_{ij}, h_{ij}\right)} \psi_{\tau_{ij}}\left(Y_{ik}- Q_{Y, h}^*(\tau_{ij}|x) \right) K\left(\frac{X_{ik}-x}{h_{ij}}\right),
    \end{equation*}
    we have $\phi_{n_i}(\tau_{ij}, h_{ij}, \D_i) = \frac{1}{n_i} \sum_{k=1}^{n_i} \xi_{ijk}$.
    It's easy to verify that
    \begin{equation*}
        \text{Var}\left(\sum_{i=1}^m \sum_{j=1}^J \frac{\omega_{ij}}{n_i} \sum_{k=1}^{n_i} \xi_{ijk}\right) = \sum_{i=1}^m  \frac{1}{n_i^2} \sum_{k=1}^{n_i} \text{Var}\left(\sum_{j=1}^J \omega_{ij}\xi_{ijk}\right),
    \end{equation*}
    where
    \begin{equation*}
        \text{Var}\left(\sum_{j=1}^J \omega_{ij}\xi_{ijk}\right) =  \frac{\sigma^{2}\left(x\right)\int_{\R} K(v)^2 dv }{f_X\left(x\right)} \B{\omega}_i^{\top} \BmR(\B{h}_i, \B{\tau}_{i}) \B{\omega}_i + o\left(\sum_{j=1}^J\sum_{j^{\prime}=1}^J \frac{\omega_{ij}\omega_{ij^{\prime}}}{\sqrt{h_{ij} h_{ij^{\prime}}}}\right).
    \end{equation*}
    Then, we have
    \begin{equation*}
        \text{Var}\left(\sum_{i=1}^m \sum_{j=1}^J \frac{\omega_{ij}}{n_i} \sum_{k=1}^{n_i} \xi_{ijk}\right) = \Sigma(x, \B{\omega}, \B{\tau}, \B{h}) + o\left(\sum_{i=1}^m \sum_{j=1}^J\frac{\omega_{ij}^2}{n_i h_{ij}}\right).
    \end{equation*}
    By Lemma~\ref{lemm_cr_sgm}, it holds that
    \begin{equation}\label{eq_sgmxwth}
        \Sigma(x, \B{\omega}, \B{\tau}, \B{h}) = \Omega\br{\|\B{\omega}\|_2^2 n^{-s\br{1-\nu}}}
    \end{equation}
    uniformly for $x$, $n$ and $m$.    
    Recalling the condition $s \in (\frac{2}{3}, 1]$ in Assumption~\ref{assum_bat_bw}, \eqref{eq_sum_rbi} leads to
    \begin{equation}\label{eq_cr_rbi}
        \sum_{i=1}^m\sum_{j=1}^J \omega_{ij} R^B_i(n_i, h) = o_p(\|\B{\omega}\|_2 n^{-\frac{1}{2}s\br{1-\nu}}).
    \end{equation}
    Then by \eqref{eq_sgmxwth} and \eqref{eq_cr_rbi}, we have
    \begin{equation*}
        \br{\sum_{i=1}^m\sum_{j=1}^J \omega_{ij} R^B_i(n_i, h)} / \sqrt{\Sigma(x, \B{\omega}, \B{\tau}, \B{h})} = o_p(1).
    \end{equation*}
    Then considering the Slutsky lemma, to prove \eqref{lim_sqrtn_nonp}, we only need to establish the asymptotic normality of
    $\br{\sum_{i=1}^m \sum_{j=1}^J \frac{\omega_{ij}}{n_i} \sum_{k=1}^{n_i} \xi_{ijk}} / \sqrt{\Sigma(x, \B{\omega}, \B{\tau}, \B{h})}$.
    Further considering the Lindeberg-Feller central limit theorem, it's sufficient to show that
    \begin{equation}\label{pf_them_nonp6}
        \sum_{i=1}^m \sum_{j=1}^J\sum_{k=1}^{n_i} \E{\frac{|\omega_{ij}|^2}{{n_i}^2} \frac{|\xi_{ijk}|^2}{\Sigma(x, \B{\omega}, \B{\tau}, \B{h})} I\left(\frac{|\omega_{ij}||\xi_{ik}|}{n_i\sqrt{\Sigma(x, \B{\omega}, \B{\tau}, \B{h})}}  >\delta \right)} \to 0, \;\; \forall \delta >0.
    \end{equation}

    By Assumption~\ref{assum_fFK}, there exists a positive constant $M_{\xi}$  independent with $i, j, k$, $m$ and $n$, such that $\left|\xi_{ijk} \right| \leq M_{\xi} < \infty$ a.s..
    By the Jensen's inequality, we have $\|\B{\omega}\|_2 \geq \br{mJ}^{-\frac{1}{2}}$, which leads to $\|\B{\omega}\|_{\infty} \|\B{\omega}\|_2^{-1} = o\br{n^{1/3}}$ recalling Assumptions~\ref{assum_wtau} and \ref{assum_bat_bw}.
    Then by Lemma~\ref{lemm_cr_sgm}, there exist a constant $M_{\Sigma}$ independent with $i, j, k$, $m$ and $n$, such that
    \begin{equation*}
        \begin{aligned}
            \frac{|\omega_{ij}|}{n_i} \frac{|\xi_{ij}|}{\sqrt{\Sigma(x, \B{\omega}, \B{h})}} &\leq M_{\Sigma} M_{\xi} \|\B{\omega}\|_{\infty} \|\B{\omega}\|_2^{-1} n^{-\frac{1}{2} s\br{1 + \nu}} \\
            &= o\br{n^{\frac{1}{3}} n^{-\frac{1}{2} s\br{1 + \nu}}} \to 0, \; a.s. \text{ as }  n \to \infty.
        \end{aligned}
    \end{equation*}
    Therefore for sufficiently large $n$, we have
    \begin{equation*}
        I\left(\frac{|\omega_{ij}||\xi_{ik}|}{n_i\sqrt{\Sigma(x, \B{\omega}, \B{\tau}, \B{h})}}  >\delta \right) = 0, \quad a.s. \quad \forall i, j, k,
    \end{equation*}
    which implies \eqref{pf_them_nonp6}.

\end{proof}

\subsection{Proof of Theorem~\ref{thm_amse}}

\begin{proof}

    The equations~\eqref{eq1_amse} and \eqref{eq1d1_amse} are direct corollaries of \eqref{lim_sqrtn_nonp}.

    By Assumption~\ref{assum_bat_bw}, it holds uniformly for $x \in [0, 1]$ that
    \begin{equation*}
        \abs{B(x, \B{\omega}, \B{\tau}, \B{h})} \leq  \sum_{i=1}^m\sum_{j=1}^J h_{ij}^2\abs{\omega_{ij}} \abs{\beta(x, \tau_{ij})} =  O\br{M_{\beta}\|\B{\omega}\|_1 n^{-2s\nu}}
    \end{equation*}
    with the constant
    \begin{equation*}
        M_{\beta} = \sup\left\{\beta(x, \tau), x \in [0, 1], \tau \in [\delta_{\tau}, 1 - \delta_{\tau}]\right\} < \infty
    \end{equation*}
    by Assumption~\ref{assum_fFK}.
    Then we have
    \begin{equation}\label{eq_amse_B}
        B^2(x, \B{\omega}, \B{\tau}, \B{h}) = O\br{\|\B{\omega}\|_1^2 n^{-4 s\nu}},
    \end{equation}
    uniformly for $x \in [0, 1]$.
    By Lemma~\ref{lemm_cr_sgm}, we have
    \begin{equation}\label{eq_amse_S}
        \Sigma(x, \B{\omega}, \B{\tau}, \B{h}) = \Omega\br{\|\B{\omega}\|_2^2 n^{-s\br{1-\nu}}},
    \end{equation}
    uniformly for $x \in [0, 1]$.
    Combining \eqref{eq_amse_B} and \eqref{eq_amse_S}, we obtain \eqref{eq2_amse}.

\end{proof}

\subsection{Proof of Theorem~\ref{thm_optw}}

\begin{proof}

    The optimal weight vector can be expressed as
    \begin{equation*}
        \B{\omega}^*\br{\B{\tau}, \B{h}} = \argmin_{\B{\omega}} \Sigma(x, \B{\omega}, \B{\tau}, \B{h}), \quad \text{ subject to } \left\{\begin{aligned}
            &\B{\omega}^{\top}\B{1}_{mJ} = 1, \\
            &\B{\omega}^{\top}\B{F}_{\varepsilon}^{-1} \br{\B{\tau}} = 0,
        \end{aligned}\right.
    \end{equation*}
    Then \eqref{eq1_optw_var} can be obtained by the Lagrange's method of multipliers.

    Let $\B{\omega}_u = 1/(mJ) \B{1}_{mJ}$ be the uniform weight vector.
    By \eqref{eq_bartau}, $\B{\omega}_u$ satisfy the constraints in \eqref{eq_nonbias}.
    Since the weight vector $\B{\omega}^* = \B{\omega}^*\br{\B{\tau}\br{\bar{\tau}^*}, \B{h}}$ is optimal, we have
    \begin{equation}\label{eq_sgm_wu}
        \Sigma\br{x, \B{\omega}^*, \B{\tau}\br{\bar{\tau}^*}, \B{h}} \leq \Sigma\br{x, \B{\omega}_u, \B{\tau}\br{\bar{\tau}^*}, \B{h}}.
    \end{equation}
    By \eqref{eq_minwoptopt} with $\B{\omega}^{**} = \B{\omega}^*\br{\B{\tau}\br{\bar{\tau}^{**}}, \B{h}}$, we also have
    \begin{equation}\label{eq_sgm_wu2}
        \Sigma\br{x, \B{\omega}^{**}, \B{\tau}\br{\bar{\tau}^{**}}, \B{h}} \leq \Sigma\br{x, \B{\omega}_u, \B{\tau}\br{\bar{\tau}^{**}}, \B{h}}.
    \end{equation}
    Then we complete the proof.
\end{proof}

\subsection{Proof of Corollary~\ref{coro_optw_norm}}

\begin{proof}
    Let $\B{\omega}^* = \B{\omega}^*\br{\B{\tau}\br{\bar{\tau}^*}, \B{h}}$.
    By Lemma~\ref{lemm_cr_sgm}, we have
    \begin{equation*}
        \Sigma\br{x, \B{\omega}^*, \B{\tau}\br{\bar{\tau}^*}, \B{h}} = \Omega\br{ n^{\nu-1}m^{1-\nu} \|\B{\omega}^*\|_2^2},
    \end{equation*}
    \begin{equation*}
        \Sigma\br{x, \B{\omega}_u, \B{\tau}\br{\bar{\tau}^*}, \B{h}} = \Omega\br{ n^{\nu-1}m^{1-\nu} \|\B{\omega}_u\|_2^2}.
    \end{equation*}
    By \eqref{eq_sgm_wu}, we have
    \begin{equation*}
        \|\B{\omega}^*\|_2^2 = O\br{\|\B{\omega}_u\|_2^2} = \Omega(m^{-1}).
    \end{equation*}
    By the Jensen's inequality, we have
    \begin{equation*}
        \|\B{\omega}^*\|_2^2 \geq \|\B{\omega}_u\|_2^2.
    \end{equation*}
    Then we obtain $\|\B{\omega}^*\|_2^2 = \Omega\br{m^{-1}}$.
    Further by the Cauchy-Schwarz inequality, we have
    \begin{equation*}
        \|\B{\omega}^*\|_1 \leq \sqrt{mJ} \|\B{\omega}^*\|_2 = O(1).
    \end{equation*}
    Then we obtain \eqref{eq_optw_norm}.

    By \eqref{eq_sgm_wu2} and similar derivations as before, \eqref{eq_optw_norm} holds for $\B{\omega} = \B{\omega}^{**}$.

\end{proof}

\end{document}